\documentclass{lmcs}
\pdfoutput=1

\usepackage{lastpage}
\lmcsdoi{18}{1}{23}
\lmcsheading{}{\pageref{LastPage}}{}{}%
{Apr.~13,~2021}{Feb.~01,~2022}{}

\usepackage[utf8]{inputenc}

\usepackage{amsmath,amssymb}
\usepackage{stmaryrd}
\usepackage{cite}

\usepackage{tikz}
\usetikzlibrary{%
  arrows,%
  arrows.meta,%
  calc,%
  chains,%
  patterns,%
  decorations.pathmorphing,%
  decorations.pathreplacing,%
  fit,%
  intersections,%
  positioning,%
  shapes.multipart,%
  svg.path,%
}

\newenvironment{quotex}{%
  \vskip\abovedisplayskip%
  \list{}{\rightmargin0pt}\item[]%
}{\endlist\vskip\belowdisplayskip%
}

\usepackage[draft,margin]{fixme}

\newcommand{\optparen}[1]{{#1}}
\newcommand{\quotparen}[1]{{#1}}
\newcommand*{\theorem}[1]{\textsc{#1}}
\newcommand*{\tc}[1]{\textit{#1}} % DT: changed from sf for consistency with F (and to differ from \const)
\newcommand*{\const}[1]{\textsf{#1}}
\newcommand*{\syntax}[1]{\textrm{\underline{#1}}}
\newcommand*{\keyw}[1]{\texttt{#1}}
\newcommand*{\var}[1]{\mathit{#1}}
\newcommand*{\hastype}{::}
\newcommand*{\fun}{\rightarrow}
\newcommand*{\relfun}{\Mapsto}
\newcommand*{\reltype}{\otimes}
\newcommand*{\unit}{\mathsf{1}}
\newcommand*{\Iff}{\longleftrightarrow}
\newcommand*{\fst}{\const{fst}}%\pi_1}
\newcommand*{\snd}{\const{snd}}%\pi_2}
\newcommand*{\relcomp}{\mathbin{\bullet}}
\newcommand*{\vimagerel}[3]{#1 \mathop{\langle} #3 \mathop{\rangle} #2}
\newcommand*{\seq}[1]{\overline{#1}}
\newcommand*{\embedding}{\mathfrak{e}}
\newcommand*{\append}{\mathbin{\cdot}}

\newcommand*{\internalize}[1]{{\kern.75pt\setlength{\fboxsep}{1pt}\fbox{$#1$}}\kern-.25pt}
\newcommand*{\internalizesim}[2]{{\kern.75pt\color{white}\setlength{\fboxsep}{1pt}\colorbox{gray}{$#1$}}\kern-.25pt}

\makeatletter
\providecommand{\leftsquigarrow}{%
  \mathrel{\mathpalette\reflect@squig\relax}%
}
\newcommand{\reflect@squig}[2]{%
  \reflectbox{$\m@th#1\rightsquigarrow$}%
}
\makeatother

\tikzstyle{circlednode}=[every node/.style={shape=circle,draw=black!70,text=black!70,inner sep=1pt,font=\small}]

\newcommand*{\circled}[1]{\smash{\tikz[baseline={(char.base)}, circlednode]{\node (char) {#1};}}}

\tikzstyle{circledxnode}=[every node/.style={shape=circle,draw=black,text=black,inner sep=0pt}]

\newcommand*{\circledx}[1]{\smash{\tikz[baseline={(char.base)}, circledxnode]{\node (char) {#1};}}}
\renewcommand{\circledast}{\circledx{$\star$}}

\newcommand*{\rtranclp}[1]{\overset{\ast}{#1}}

\newcommand{\exampleend}{%
  \bgroup
  \let\oldqedsymbol=\qedsymbol
  \def\qedsymbol{\smash{\scalebox{0.8}{\rotatebox{45}{\oldqedsymbol}}}}%
  \qed
  \egroup
}

\title{Quotients of Bounded Natural Functors}

\author[B.~F\"urer]{Basil F\"urer\rsuper{a}}
\address{Department of Computer Science, ETH Z\"urich, Switzerland}
\author[A.~Lochbihler]{Andreas Lochbihler\rsuper{b}}
\address{Digital Asset (Switzerland) GmbH, Zurich, Switzerland}
\email{mail@andreas-lochbihler.de}
\author[J.~Schneider]{Joshua Schneider\rsuper{c}}
\address{Institute of Information Security, Department of Computer Science, ETH Z\"urich, Switzerland}
\email{joshua.schneider@inf.ethz.ch}
\author[D.~Traytel]{Dmitriy Traytel\rsuper{d}}
\address{Department of Computer Science, University of Copenhagen, Denmark}
\email{traytel@di.ku.dk}


\begin{document}

%
%\authorrunning{B. F\"urer et al.}
%
%\institute{
%  Institute of Information Security, Department of Computer Science, ETH Z\"urich, Switzerland%\\
%  %\email{\{joshua.schneider,traytel\}@inf.ethz.ch}
%  \and
%  Digital Asset (Switzerland) GmbH, Zurich, Switzerland%\\
%  %\email{mail@andreas-lochbihler.de}
%}

\maketitle              % typeset the header of the contribution

\begin{abstract}
%   \looseness=-1
  The functorial structure of type constructors is the foundation for many definition and proof principles in higher-order logic (HOL).
  For example, inductive and coinductive datatypes can be built modularly from bounded natural functors (BNFs), a class of well-behaved type constructors.
  Composition, fixpoints, and---under certain conditions---subtypes are known to preserve the BNF structure.
  In this article, we tackle the preservation question for quotients, the last
  important principle for introducing new types in HOL\@.
  We identify sufficient conditions under which a quotient inherits the BNF structure from its underlying type.
  Surprisingly, lifting the structure in the obvious manner fails for some quotients, a problem that also affects the quotients of polynomial functors used in the Lean proof assistant.
  We provide a strictly more general lifting scheme that supports such problematic quotients.
  We extend the Isabelle/HOL proof assistant with a command that automates the registration of a quotient type as a BNF, %by lifting the underlying type's BNF structure.
  reducing the proof burden on the user from the full set of BNF axioms to our inheritance conditions.
  We demonstrate the command's usefulness through several case studies.
\end{abstract}

\section{Introduction}

The functorial structure of type constructors forms the basis for many definition and proof principles in proof assistants.
Examples include datatype and codatatype definitions~\cite{AvigadCarneiroHudon2019ITP,blanchette14itp,TraytelPopescuBlanchette2012LICS}, program synthesis~\cite{CohenDenesMortberg2013CPP,HuffmanKuncar2013CPP,LammichLochbihler2018JAR}, generalized term rewriting~\cite{Sozeau2010JFR}, and reasoning based on representation independence~\cite{BasinLochbihlerSefidgar2020JC,HuffmanKuncar2013CPP,KuncarPopescu2019JAR} and about effects~\cite{Lochbihler2019jar,LochbihlerSchneider2016ITP}.

A type constructor becomes a functor through a mapper operation that lifts functions on the type arguments to the constructed type.
The mapper must be \emph{functorial}, i.e., preserve identity functions ($\const{id}$) and distribute over function composition ($\circ$).
For example, the list type constructor $\_\;\tc{list}$%
\footnote{Type constructors are written postfix in this article.}
has the well-known mapper $\const{map} \hastype (\alpha\fun\beta) \fun \alpha\;\tc{list} \fun \beta\;\tc{list}$, which applies the given function to every element in the given list. It is functorial:
\begin{equation*}
  \const{map}\;\const{id} = \const{id}
  \qquad\qquad
  \const{map}\;g \circ \const{map}\;f = \const{map}\;(g \circ f)
\end{equation*}

%\looseness=-1
Most applications of functors can benefit from even richer structures.
In this article, we focus on bounded natural functors (BNFs)~\cite{TraytelPopescuBlanchette2012LICS}.
A~BNF comes with additional setter operators that return sets of occurring elements, called atoms, for each type argument (Section~\ref{section:BNF}).
The setters must be \emph{natural} transformations, i.e., commute with the mapper, and \emph{bounded}, i.e., have a fixed cardinality bound on the sets they return.
For example, lists are a BNF with the setter $\const{set} \hastype \alpha\;\tc{list} \fun \alpha\;\tc{set}$, which returns the set of elements in a list. It satisfies $\const{set}\circ\const{map}\;f = f\langle \_\rangle \circ \const{set}$, where $f\langle \_\rangle$ denotes the function that maps a set $X$ to $f\langle X\rangle = \{f\;x\mid x \in X\}$, i.e., the image of $X$ under $f$.
Moreover, $\const{set}\;\var{xs}$ is always a finite set because lists are finite sequences.

Originally, BNFs were introduced for modularly constructing datatypes and codata\-types~\cite{blanchette14itp} in the Isabelle/HOL proof assistant.
Although (co)datatypes are still the most important use case, the BNF structure is used nowadays in other contexts such as reasoning via free theorems~\cite{LochbihlerSefidgarBasinMaurer2019CSF} and transferring theorems between types~\cite{Kuncar2016PhD,LochbihlerSchneider2018ITP}.
%\pagebreak[3]

\looseness=-1
Several type definition principles in HOL
preserve the BNF structure:
composition (e.g., $(\alpha\;\tc{list})\;\tc{list}$), %least and greatest fixpoints
datatypes and codatatypes~\cite{TraytelPopescuBlanchette2012LICS}, and---under certain conditions---subtypes~\cite{Biendarra2015BA,LochbihlerSchneider2018ITP}.
Subtypes include records and type copies.
Accordingly, when a new type constructor is defined via one of these principles from an existing BNF, then the new type automatically comes with a mapper and setters and with theorems for the BNF properties.

One important type definition principle is missing above: quotients~\cite{Homeier2005TPHOLs,HuffmanKuncar2013CPP,KaliszykUrban2011SAC,Paulson2006TCL,Slotosch1997TPHOLs}.
A quotient type (Section~\ref{section:quotient}) identifies elements of an underlying type according to a (partial) equivalence relation $\optparen{\sim}$.
That is, the quotient type is isomorphic to the equivalence classes of $\optparen{\sim}$.
For example, unordered pairs $\alpha\;\tc{upair}$ are the quotient of ordered pairs $\alpha \times \alpha$ and the equivalence relation $\optparen{\sim_{\tc{upair}}}$ generated by $(x, y) \sim_{\tc{upair}} (y, x)$.
Similarly, finite sets, bags, and cyclic lists are quotients of lists where the equivalence relation permutes or duplicates the list elements as needed.

In this article, we answer the question when and how a quotient type inherits its underlying type's BNF structure.
It is well known that a quotient preserves the functorial properties if the underlying type's mapper preserves $\optparen{\sim}$;
then the quotient type's mapper is simply the lifting of the underlying type's mapper to equivalence classes~\cite{AvigadCarneiroHudon2019ITP}.

For setters, the situation is more complicated.
Ad\'amek et al.~\cite{AdamekGummTrnkova2010JLC} call a functor \emph{sound} if it preserves empty intersections.
All BNFs are sound.
However, we discovered that if the setters are defined as one would expect for a quotient, then the resulting functor may be unsound.
To repair the situation, we characterize the setters in terms of the mapper and identify a definition scheme for the setters that results in sound functors.
We then derive sufficient conditions on the equivalence relation $\optparen{\sim}$ for the BNF properties to be preserved for these definitions (Section~\ref{section:theory}).
%\begin{conf}
%With few exceptions, we omit proofs and refer to our technical report~\cite{report}, which contains them.
%\end{conf}

%\looseness=-1
Moreover, we have implemented an Isabelle/HOL command that automates the registration of a quotient type as a BNF (Section~\ref{section:implementation});
the user merely needs to discharge the conditions on $\optparen{\sim}$.
One of the conditions, subdistributivity, often requires considerable proof effort, though.
We therefore developed a novel sufficient criterion using confluent relations that simplifies the proofs in our case studies (Section~\ref{section:confluence}).
Our implementation is distributed with the Isabelle2021 release. Some of the examples in this article are only available in Isabelle's development repository\footnote{\href{https://isabelle.in.tum.de/repos/isabelle}{\texttt{https://isabelle.in.tum.de/repos/isabelle}}, revision be11fe268b33} and will be part of the forthcoming Isabelle release.
%available in Isabelle's development repository\footnote{\href{https://isabelle.in.tum.de/repos/isabelle}{\texttt{https://isabelle.in.tum.de/repos/isabelle}}, revision fce780f9c9c6} and will be part of the forthcoming Isabelle release.

\paragraph{Contributions}
The main contributions of this article are the following:
\begin{enumerate}
\item
  We identify sufficient criteria for when a quotient type preserves the BNF properties of the underlying type.
  Registering a quotient as a BNFs allows (co)datatypes to nest recursion through it. Consider for example node-labeled unordered binary trees
  \begin{quotex}
      $\keyw{datatype}\;\tc{ubtree} = \const{Leaf} \mid \const{Node}\;\tc{nat}\; (\mkern-2mu\tc{ubtree}\;\tc{upair})$
  \end{quotex}
  BNF use cases beyond datatypes benefit equally.

\item %\looseness=-1
  In particular, we show that the straightforward definitions would cause the functor to be unsound, and find better definitions that avoid unsoundness.
  This problem is not limited to BNFs.
  The lifting operations for Lean's quotients of polynomial functors~\cite{AvigadCarneiroHudon2019ITP} also suffer from unsoundness and our repair applies to them as well (Section~\ref{section:QPF}). We show in Section~\ref{section:QPF} that unsoundness hinders modular proofs.

\item
  We propose a sufficient criterion on $\optparen{\sim}$ for subdistributivity,
  which is typically the most difficult BNF property to show.
  We show with several examples that the criterion is applicable in practice and yields relatively simple proofs.

\item
%\looseness=-1
  We have implemented an Isabelle/HOL command to register the quotient as a BNF once the user has discharged the conditions on $\optparen{\sim}$.
  The command also generates proof rules for transferring theorems about the BNF operations from the underlying type to the quotient (Section~\ref{section:transfer}).
  Several case studies demonstrate the command's usefulness.
  Some examples reformulate well-known BNFs as quotients (e.g., unordered pairs, distinct lists, finite sets). %; this allows us to estimate the effort savings. \fxnote{DT: I commented subsection 4.2, which was supposed to do this. I don't think a whole subsection is justified. Maybe we can add such things as notes to the examples? Or drop the sentence after semicolon.}
  Others formally prove the BNF properties for the first time, e.g., cyclic lists, the free idempotent monoid, and regular expressions modulo ACI\@.
  These examples become part of the collection of formalized BNFs and can thus be used in datatype definitions and other BNF applications.
\end{enumerate}

\begin{exa}\label{ex:regex}
    \looseness=-1
    To illustrate our contributions' usefulness,
    we consider linear dynamic logic (LDL)~\cite{DBLP:conf/ijcai/GiacomoV13}, an extension of linear temporal logic with regular expressions.
    LDL's syntax is usually given as two mutually recursive datatypes of formulas and regular expressions~\cite{DBLP:conf/ijcai/GiacomoV13,DBLP:conf/rv/BasinKT17}.
    Here, we opt for nested recursion, which has the modularity benefit of being able to formalize regular expressions separately.
    We define regular expressions $\alpha\;\tc{re}$:% with tests represented by $\alpha$.
    \begin{quotex}
        \begin{tabular}{@{}r@{\,}c@{}l@{}}
        \keyw{datatype}\;$\alpha\;\tc{re}$&${}={}$&$\const{Zero} \mid \const{Eps} \mid \const{Atom}\;\alpha$\\&$\mid$&$ \const{Alt}\;(\alpha\;\tc{re})\;(\alpha\;\tc{re}) \mid \const{Conc}\;(\alpha\;\tc{re})\;(\alpha\;\tc{re}) \mid \const{Star}\;(\alpha\;\tc{re})$
    \end{tabular}
    \end{quotex}

    Often, it is useful to consider regular expressions modulo some syntactic equivalences.
    For example, identifying expressions modulo the associativity, commutativity, and idempotence (ACI) of the alternation constructor $\const{Alt}$ results in a straightforward construction of deterministic finite automata from regular expressions via Brzozowski derivatives~\cite{DBLP:conf/itp/NipkowT14}. We define the ACI-equivalence $\optparen{\sim_\tc{aci}}$ as the least congruence relation satisfying:
    \begin{quotex}
        $\const{Alt}\;(\const{Alt}\;r\;s)\;t \sim_\tc{aci} \const{Alt}\;r\;(\const{Alt}\;s\;t)$
        \quad\;
        $\const{Alt}\;r\;s \sim_\tc{aci} \const{Alt}\;s\;r$
        \quad\;
        $\const{Alt}\;r\;r \sim_\tc{aci} r$
    \end{quotex}
%
    %$r\sim_\tc{aci}r'\implies s\sim_\tc{aci}s'\implies \const{Alt}\;r\;s \sim_\tc{aci} \const{Alt}\;r'\;s'$
    %
    %$r\sim_\tc{aci}r'\implies s\sim_\tc{aci}s'\implies \const{Conc}\;r\;s \sim_\tc{aci} \const{Conc}\;r'\;s'$
    %
    %$r\sim_\tc{aci}r'\implies \const{Star}\;r \sim_\tc{aci} \const{Star}\;r'$
%
    Next, we define the quotient type of regular expressions modulo ACI $\alpha\;\tc{re}_\tc{aci}$ and the datatype of LDL formulas $\tc{ldl}$, which uses nested recursion through $\alpha\;\tc{re}_\tc{aci}$.
    \begin{quotex}
        \keyw{quotient\_type}\;$\alpha\;\tc{re}_\tc{aci} = \alpha\;\tc{re} / \quotparen{\sim_\tc{aci}}$\\[0.7\jot]
        \keyw{datatype}\;$\tc{ldl} = \const{Prop}\;\tc{string} \mid \const{Neg}\;\tc{ldl} \mid \const{Conj}\;\tc{ldl}\;\tc{ldl} \mid \const{Match}\;(\tc{ldl}\;\tc{re}_\tc{aci})$
    \end{quotex}
    For the last declaration to succeed, Isabelle must know that $\alpha\;\tc{re}_\tc{aci}$ is a BNF\@. We will show in Section~\ref{section:confluence} how our work allows us to lift $\alpha\;\tc{re}$'s BNF structure to $\alpha\;\tc{re}_\tc{aci}$.\exampleend
\end{exa}

This article extends and revises the homonymous IJCAR conference
paper~\cite{DBLP:conf/cade/FurerLST20}. Specifically, the article newly
describes the interaction of quotients and \emph{non-emptiness
    witnesses}~\cite{DBLP:conf/esop/Blanchette0T15}, an additional piece of
information tracked as part of the BNF structure and used to prove non-emptiness
of inductive datatypes, which is a requirement for introducing new types in HOL
(Section~\ref{section:witness}). We also show how to lift the BNF structure to
\emph{partial quotients} by combining our constructions for quotients with the
ones for subtypes, and discuss limitations of this approach
(Section~\ref{section:partial:quotient}). Moreover, we include several previously omitted pen-and-paper proofs of our results, whose formalized counterparts are implemented as tactics as part of our Isabelle command to automate the lifting. We also give a more detailed
description of our command's interface
(Section~\ref{section:command}) and formalize several new examples, some of which required us to extend our results' scope. Notably, the new Example~\ref{ex:regex2} uses the new Lemma~\ref{lemma:wide:intersection:finite} and the updated Theorem~\ref{thm:confluent:quotient}, which generalizes the corresponding Theorem~4 from the conference paper.

\section{Background}

%Isabelle is a generic proof assistant that supports multiple logical foundations.
We work in Isabelle/HOL, Isabelle's variant of classical higher-order logic---a simply typed theory with Hilbert choice and rank-1 polymorphism. %\fxnote{DT: I don't mention type classes because they are not relevant here, ok?}
We refer to a textbook for a detailed introduction to Isabelle/HOL~\cite{DBLP:books/sp/NipkowK14} and only summarize relevant notation here.

\looseness=-1
Types are built from type variables $\alpha,\,\beta,\,\ldots$ via type constructors.
A type constructor can be nullary (\tc{nat}) or have some type arguments ($\alpha\;\tc{list}$, $\alpha\;\tc{set}$, $(\alpha,\,\beta)\;\tc{upair}$).
Type constructor application is written postfix.
Exceptions are the binary type constructors for sums ($+$), products ($\times$), and functions ($\fun$), all written infix.
Terms are built from variables $x,\,y,\,\ldots$ and constants $\const{c},\,\const{d},\,\ldots$ via lambda-abstractions $\lambda x.\; t$ and applications $t\;u$. The sum type's embeddings are $\const{Inl}$ and $\const{Inr}$ and the product type's projections are $\fst$ and $\snd$.

\looseness=-1
The primitive way of introducing new types in HOL is to take a non-empty subset of an existing type.
For example, the type of lists could be defined as the set of pairs $(n \hastype \tc{nat},\,f \hastype \tc{nat} \fun \alpha)$ where $n$ is the list's length and $f\;i$ is the list's $i$th element for $i < n$ and some fixed unspecified element of type $\alpha$ for $i \geq n$. %\pagebreak[2]
To spare the users from such low-level encodings, Isabelle/HOL offers higher-level mechanisms for introducing new types, which are internally reduced to primitive subtyping.
In fact, lists are defined as an inductive  $\keyw{datatype}\;\alpha\;\tc{list} = [] \mid \alpha \mathbin\# \alpha\;\tc{list}$, where $[]$ is the empty list and $\#$ is the infix list constructor.
Recursion in datatypes and their coinductive counterparts may take place only under well-behaved type constructors, the bounded natural functors~(Section~\ref{section:BNF}).
Quotient types (Section~\ref{section:quotient}) are another high-level mechanism for introducing new types.

For $n$-ary definitions, we use the vector notation $\seq{x}$ to denote the sequence $x_1, \ldots, x_n$, where $n$ is clear from the context.
Vectors spanning several variables indicate repetition with synchronized indices. For example, $\const{map}_F\;\seq{(g\circ f)}$ abbreviates $\const{map}_F\;({g_1\circ f_1})\;\ldots\;({g_n\circ f_n})$. Abusing notation slightly, we write $\seq{\alpha} \fun \beta$ for the $n$-ary function type $\alpha_1 \fun \cdots \fun \alpha_n \fun \beta$.

\looseness=-1
To simplify notation, we identify the type of binary predicates $\alpha \fun \beta \fun \tc{bool}$ and sets of pairs $(\alpha \times \beta)\;\tc{set}$, and write $\alpha \otimes \beta$ for both.
These types are different in Isabelle/HOL and the BNF ecosystem works with binary predicates.
The identification allows us to mix set and relation operations, e.g., the subset relation $\subseteq$ and relation composition $\relcomp$ (both written infix).

\subsection{Bounded Natural Functors}%
\label{section:BNF}

\looseness=-1
A bounded natural functor (BNF)~\cite{TraytelPopescuBlanchette2012LICS} is an $n$-ary type constructor $\seq{\alpha}\;F$ equipped with the following polymorphic constants. Here and elsewhere, $i$ implicitly ranges over $\{1,\,\ldots,\,n\}$:
\begin{equation*}
\begin{tabular}{@{}l@{\qquad\qquad}l@{}}
$\const{map}_F \hastype \seq{(\alpha \fun \beta)} \fun \seq{\alpha}\;F \fun \seq{\beta}\;F$&

$\const{bd}_{F} \hastype \tc{cardinal}_F$\\[0.7\jot]

$\const{set}_{F,i} \hastype \seq{\alpha}\;F \fun \alpha_i\;\tc{set}$\quad for all $i$&

$\const{rel}_F \hastype \seq{(\alpha \reltype \beta)} \fun \seq{\alpha}\;F \reltype \seq{\beta}\;F$

\end{tabular}
\end{equation*}

%\pagebreak[2]

\looseness=-1
The \emph{shape and content} intuition~\cite{TraytelPopescuBlanchette2012LICS} is a useful way of thinking about elements of $\seq{\alpha}\;F$.
The mapper $\const{map}_F$ leaves the shape unchanged but modifies the contents by applying its function arguments.
The $n$ setters $\const{set}_{F,i}$ extract the contents (and dispose of the shape).
For example, the shape of a list is given by its length, which $\const{map}$ preserves.
%The setter $\const{set} \hastype \alpha\;\tc{list} \fun \alpha\;\tc{set}$ returns the elements of a list as a set.
The cardinal bound $\const{bd}_F$ is a fixed bound on the number of elements returned by $\const{set}_{F,i}$.
Cardinal numbers are represented in HOL using particular well-ordered relations~\cite{DBLP:conf/itp/Blanchette0T14} over a large-enough type (specific to $F$).
We mention the bounds and cardinal numbers only for completeness;
they are not interesting for this article.
Finally, the relator $\const{rel}_F$ lifts relations on the type arguments to a relation on $\seq{\alpha}\;F$ and $\seq{\beta}\;F$. Thereby, it only relates elements of $\seq{\alpha}\;F$ and $\seq{\beta}\;F$ that have the same shape.

The BNF constants must satisfy the following properties:

\begin{center}
\begin{tabular}{@{}r@{\;\;}l@{\qquad}l@{\;}l@{}}
\theorem{map\_id}&$\const{map}_F\;\seq{\const{id}} = \const{id}$\\[0.7\jot]

\theorem{map\_comp}&$\const{map}_F\;\seq{g} \circ \const{map}_F\;\seq{f} = \const{map}_F\;\seq{(g \circ f)}$\\[0.7\jot]

\theorem{set\_map}&$\const{set}_{F,i} \circ \const{map}_F\;\seq{f} = f_i\langle \_\rangle \circ \const{set}_{F,i}$\\[0.7\jot]

\theorem{map\_cong}&\multicolumn{3}{@{}l@{}}{$(\forall i.\; \forall z \in \const{set}_{F,i}\;x.\;f_i\;z = g_i\;z)\implies \const{map}_F\;\seq{f}\;x = \const{map}_F\;\seq{g}\;x$}\\[0.7\jot]

\theorem{set\_bd}&$|\const{set}_{F,i}\;x| \leq_o \const{bd}_F$\\[0.7\jot]

\theorem{bd}&$\const{infinite\_card}\;\const{bd}_F$\\[0.7\jot]

\theorem{in\_rel}&\multicolumn{3}{@{}l@{}}{$\const{rel}_F\;\seq{R}\;x\;y = \exists z.\;(\forall i.\;\const{set}_{F,i}\;z\subseteq R_i) \land \const{map}\;\seq{\fst}\;z = x \land \const{map}\;\seq{\snd}\;z = y$}\\[0.7\jot]

\theorem{rel\_comp}&$\const{rel}_F\;\seq{R} \relcomp \const{rel}_F\;\seq{S} \subseteq \const{rel}_F\;\seq{(R \relcomp S)}$
\end{tabular}
\end{center}
%
%The quantification over $i$ is a meta-notation, which abbreviates an $n$-ary conjunction.
%
\looseness=-1
Properties \theorem{map\_id} and \theorem{map\_comp} capture the mapper's functoriality; \theorem{set\_map} the setters' naturality. Moreover, the mapper and the setters must agree on what they identify as content (\theorem{map\_cong}).
Any set returned by $\const{set}_{F,i}$ must be bounded (\theorem{set\_bd}); the operator $\leq_o$ compares cardinal numbers~\cite{DBLP:conf/itp/Blanchette0T14}.
The bound is required to be infinite (\theorem{bd}), which simplifies arithmetics.
The relator can be expressed in terms of the mapper and the setter (\theorem{in\_rel}) and must distribute over relation composition (\theorem{rel\_comp}). %\pagebreak[2]
The other inclusion, namely $\const{rel}_F\;\seq{(R \relcomp S)} \subseteq \const{rel}_F\;\seq{R} \relcomp \const{rel}_F\;\seq{S}$, follows from these properties.
We refer to \theorem{rel\_comp} as \emph{subdistributivity} because it only requires one inclusion. In principle, the setter can also be expressed in terms of the mapper as the least set satisfying the congruence rule \theorem{map\_cong}. We rely on this fact in Section~\ref{section:theory}. Making the setters and the relator part of the BNF structure (rather than defining everything from the mapper) simplifies the statement of the BNF properties.

A useful derived operator is the action on sets $\internalize{F} \hastype \seq{\alpha\;\tc{set}} \fun \seq{\alpha}\;F\;\tc{set}$, which generalizes the type constructor's action on its type arguments.
Formally, $\internalize{F}\;\seq{A} = \{x \mid \forall i.\; \const{set}_{F,i}\;x\subseteq A_i\}$. Note that we can write $z \in \internalize{F}\;\seq{R}$ to replace the equivalent $\forall i.\;\const{set}_{F,i}\;z\subseteq R_i$ in \theorem{in\_rel}.
%The operator is monotone, i.e., $\internalize{F}\;\seq{A} \subseteq \internalize{F}\;\seq{B}$ if $A_i \subseteq B_i$ for all $i$. \fxnote{DT: not sure if useful.}
%Moreover, we can rewrite property \theorem{in\_rel} using $\internalize{F}$ as $\const{rel}_F\;\seq{R}\;x\;y = \exists z \in \internalize{F}\;\seq{R}.\;\const{map}\;\seq{\fst}\;z = x \land \const{map}\;\seq{\snd}\;z = y$.

\looseness=-1
Most basic types are BNFs, notably, sum and product types.
BNFs are closed under composition, e.g., $\unit + \alpha \times \beta$ is a BNF with the mapper $\lambda f\; g.\;\const{map}_{\unit +}\;(\const{map}_\times\;f\;g)$, where $\unit$ is the unit type (consisting of the single element $\star$) and $\const{map}_{\unit +}\;h = \const{map}_{+}\;\const{id}\;h$.
Moreover, BNFs support fixpoint operations, which correspond to (co)datatypes,  and are closed under them~\cite{TraytelPopescuBlanchette2012LICS}. For instance, the \keyw{datatype} command internally computes a least solution for the fixpoint type equation $\beta = \unit + \alpha \times \beta$ %, where the right-hand-side must be a BNF,
to define the $\alpha\;\tc{list}$ type.
Closure means that the resulting datatype, here $\alpha\;\tc{list}$, is equipped with the BNF structure, specifically the mapper $\const{map}$.
Also subtypes inherit the BNF structure under certain
%\begin{conf}%
%conditions~\cite{Biendarra2015BA}.
%\end{conf}
%\begin{rep}%
conditions~(Section~\ref{section:partial:quotient}).
%\end{rep}
For example, the subtype $\alpha\;\tc{nelist}$ of non-empty lists $\{\var{xs} :: \alpha\;\tc{list} \mid \var{xs} \neq []\}$ is a BNF\@.

\subsection{Quotient types}%
\label{section:quotient}

\looseness=-1
An equivalence relation $\optparen{\sim}$ on a type $T$ partitions the type into equivalence classes.
Isa\-belle/HOL supports the definition of the quotient type $Q = T / \quotparen{\sim}$, which yields a new type $Q$ isomorphic to the set of equivalence classes~\cite{KaliszykUrban2011SAC}.
For example, consider $\optparen{\sim_{\tc{fset}}}$ that relates two lists if they have the same set of elements, i.e., $\var{xs} \sim_{\tc{fset}} \var{ys}$ iff $\const{set}\;\var{xs} = \const{set}\;\var{ys}$.
The following command defines the type $\alpha\;\tc{fset}$ of finite sets as a quotient of lists:
\begin{quotex}
  $\keyw{quotient\_type }\alpha\;\tc{fset} = \alpha\;\tc{list} / \quotparen{\sim_{\tc{fset}}}$
\end{quotex}
This command requires a proof that $\optparen{\sim_{\tc{fset}}}$ is, in fact, an equivalence relation.
The relationship between a quotient type $Q$ and the underlying type $T$ is formally captured by the correspondence relation $\const{cr}_{Q} \hastype T \otimes Q$.
For example, $(\mathit{xs}, X) \in \const{cr}_{\tc{fset}}$ iff the list $\mathit{xs}$ is a representative of the finite set $X$, i.e.,
$X$ corresponds to the unique equivalence class that contains $\mathit{xs}$.

The Lifting and Transfer tools~\cite{HuffmanKuncar2013CPP,Kuncar2016PhD} automate the lifting of definitions and theorems from the raw type $T$ to the quotient $Q$.
For example, the image operation on finite sets can be obtained by lifting the list mapper $\const{map}$ using the command
\begin{quotex}
  $\keyw{lift\_definition }\const{fimage} \hastype (\alpha \fun \beta) \fun \alpha\;\tc{fset} \fun \beta\;\tc{fset}\keyw{ is }\const{map}$
\end{quotex}
Lifting is only possible for terms that respect the quotient.
For $\const{fimage}$, respectfulness states that $\const{map}\;f\;\var{xs} \sim_{\tc{fset}} \const{map}\;f\;\var{ys}$ whenever $\var{xs} \sim_{\tc{fset}} \var{ys}$.
%A proof of this property must be supplied by the user.

Lifting and Transfer are based on \emph{transfer rules} that relate two terms of possibly different types. % and guide the lifting.
The \keyw{lift\_definition} command automatically proves the transfer rule
\begin{equation*}
  (\const{map},\const{fimage}) \in ((=) \relfun \const{cr}_{\tc{fset}} \relfun \const{cr}_{\tc{fset}})
\end{equation*}
where $R \relfun S$ (right-associative) relates two functions iff they map $R$-related arguments to $S$-related results.
%The correspondence relation $\const{cr}_{\tc{fset}}$ relates a list with the finite set that it represents, i.e., the set whose corresponding equivalence class contains the list.
The meaning of the above rule is that applying $\const{map}\;f$ to a list representing the finite set $X$ results in a list that represents $\const{fimage}\;f\;X$, for all $f$.
The transfer rule's relation $(=) \relfun \const{cr}_{\tc{fset}} \relfun \const{cr}_{\tc{fset}}$ is constructed according to the types of the related terms.
This enables the composition of transfer rules to relate larger terms.
For instance, the Transfer tool derives the following equivalence using the above and other transfer rules:
\begin{equation*}
   (\forall \mathit{xs}.\;\const{set}\;(\const{map}\;\const{id}\;\mathit{xs}) = \const{set}\;\mathit{xs}) \Iff (\forall X.\;\const{fimage}\;\const{id}\;X = X)
\end{equation*}
Thus, one can prove the equation $\forall X.\;\const{fimage}\;\const{id}\;X = X$ by reasoning about lists.

Proper equivalence relations are reflexive.
Therefore, every element of the type $T$ is part of exactly one equivalence class.
It is also possible to define a partial quotient from a partial equivalence relation, which might not be reflexive.
The \keyw{quotient\_type} command and the Lifting and Transfer tools support partial quotients.
Note that partial quotients subsume subtypes (take the restriction of equality to the subset as the partial equivalence).

\section{Quotients of Bounded Natural Functors}%
\label{section:theory}

We develop the theory for when a quotient type inherits the underlying type's BNF structure.
We consider the quotient $\seq{\alpha}\;Q = \seq{\alpha}\;F / \quotparen{\sim}$ of an $n$-ary BNF $\seq{\alpha}\;F$ over an equivalence relation $\optparen{\sim}$ on $\seq{\alpha}\;F$.
The first idea is to define $\const{map}_Q$ and $\const{set}_{Q,i}$ in terms of $F$'s operations:
\begin{quotex}
  $\keyw{quotient\_type}\;\seq{\alpha}\;Q = \seq{\alpha}\;F / \quotparen{\sim}$
  \\[0.7\jot]
  $\keyw{lift\_definition}\;\const{map}_Q \hastype \seq{(\alpha \fun \beta)} \fun \seq{\alpha}\;Q \fun \seq{\beta}\;Q\;\;\keyw{is}\;\;\const{map}_F$
  \\[0.7\jot]
  $\keyw{lift\_definition}\;\const{set}_{Q,i} \hastype \seq{\alpha}\;Q \fun \alpha_i\;\tc{set}\;\;\keyw{is}\;\;\const{set}_{F,i}$
\end{quotex}
These three commands require the user to discharge the following proof obligations:
\begin{equation}
  \label{eq:equiv:sim}
  \const{equivp}\ \optparen{\sim}
\end{equation}
\begin{equation}
  \label{eq:map:respect}
  x \sim y \Longrightarrow \const{map}_F\;\seq{f}\;x \sim \const{map}_F\;\seq{f}\;y
\end{equation}
\begin{equation}
  \label{eq:set:respect:naive}
  x \sim y \Longrightarrow \const{set}_{F,i}\;x = \const{set}_{F,i}\; y
\end{equation}
\looseness=-1
The first two conditions are as expected:
$\optparen{\sim}$ must be an equivalence relation, by~\eqref{eq:equiv:sim}, and compatible with $F$'s mapper, by~\eqref{eq:map:respect}, i.e., $\const{map}_F$ preserves $\optparen{\sim}$.
The third condition, however, demands that equivalent values contain the same atoms.
This rules out many practical examples including the following simplified (and therefore slightly artificial) one.

\begin{exa}\label{ex:Inl:Inl}
    Consider $\alpha\;F_P = \alpha + \alpha$ with the equivalence relation $\optparen{\sim_P}$ generated by $\const{Inl}\;x \sim_P \const{Inl}\;y$, where $\const{Inl}$ is the sum type's left embedding. %\fxerror{DT: JCB suggests $'a + 'b$. \\ AL: No, then an argument further down will no longer work. Because then Inr is suddenly a witness that does not need $'a$.}
    That is, $\optparen{\sim_P}$ identifies all values of the form $\const{Inl}\;z$ and thus $\alpha\;Q_P = \alpha\;F_P / \quotparen{\sim_P}$ is isomorphic to the type $\unit + \alpha$. However, $\const{Inl}\;x$ and $\const{Inl}\;y$ have different sets of atoms $\{x\}$ and $\{y\}$, assuming $x\neq y$.
    \exampleend
\end{exa}

We now derive better definitions for the setters and conditions under which they preserve the BNF properties.
To that end, we characterize setters in terms of the mapper (Section~\ref{section:setter:characterization}).
Using this characterization, we derive the relationship between $\const{set}_{Q,i}$ and $\const{set}_{F,i}$ and identify the conditions on $\optparen{\sim}$ (Section~\ref{section:repair}).
Next, we do the same for the relator (Section~\ref{section:relator}).
We thus obtain the conditions under which $\seq{\alpha}\;Q$ preserves $F$'s BNF properties.
%\begin{rep}

%\end{rep}
\looseness=-1
One of the conditions, the relator's subdistributivity over relation composition, is often difficult to show directly in practice.
We therefore present an easier-to-establish criterion for the special case where a confluent rewrite relation $\optparen{\rightsquigarrow}$ generates $\optparen{\sim}$ (Section~\ref{section:confluence}).

Finally, we discuss the interaction of quotients with non-emptiness witnesses (Section~\ref{section:witness}), an additional piece of information tracked by BNFs, and the generalization to partial quotients, where $\optparen{\sim}$ is a partial equivalence relation, i.e., %symmetric and transitive, but
not necessarily reflexive (Section~\ref{section:partial:quotient}).

\subsection{Characterization of the BNF setter}%
\label{section:setter:characterization}

We now characterize $\const{set}_{F,i}$ in terms of $\const{map}_F$ for an arbitrary BNF $\seq{\alpha}\;F$.
%We first look at $F$'s action $\internalize{F}$ on sets.
Observe that $F$'s action  $\internalize{F}\;\seq{A}$ on sets contains all values that can be built with atoms from $\seq{A}$.
Hence, $\const{set}_{F,i}\;x$ is the smallest set $A_i$ such that $x$ can be built from atoms in $A_i$. Formally, the next equation follows directly from the definition of \internalize{F}:
\begin{equation}
  \label{eq:set:in}
  \const{set}_{F,i}\;x = \bigcap\{A_i \mid x \in \internalize{F}\;\seq{\const{UNIV}}\;A_i\;\seq{\const{UNIV}}\}
\end{equation}
\looseness=-1
Only atoms of type $\alpha_i$ are restricted; all atoms of other types $\alpha_j$ may come from $\const{UNIV}$, the set of all elements of type $\alpha_j$.
Moreover, $\internalize{F}$ can be defined without $\const{set}_{F,i}$, namely by trying to distinguish values using the mapper.
Informally, $x$ contains atoms not from $\seq{A}$ iff $\const{map}_F\;\seq{f}\;x$ differs from $\const{map}_F\;\seq{g}\;x$ for some functions $\seq{f}$ and $\seq{g}$ that agree on $\seq{A}$.
Hence, we obtain
\begin{equation}
  \label{eq:F_in}
  \internalize{F}\;\seq{A} = \{x \mid \forall \seq{f}\;\seq{g}.\;(\forall i.\;\forall a\in A_i.\; f_i\; a = g_i\; a) \longrightarrow \const{map}_F\;\seq{f}\;x = \const{map}_F\;\seq{g}\;x\},
\end{equation}
where $f_i, g_i \hastype \alpha_i \fun \unit + \alpha_i$.
The range type $\unit + \alpha_i$ adds a new atom $\circledast = \const{Inl}\;\star$ to the atoms of type $\alpha_i$.
Thus $\unit + \alpha_i$ contains at least two atoms, as all HOL types are inhabited, and $f_i$ and $g_i$ can therefore meaningfully distinguish atoms (for singleton types $\seq{\alpha}$, the right hand side would hold trivially if $f_i, g_i$ had type the $\alpha_i \fun \alpha_i$ because there is only one such function).
We write $\embedding \hastype \alpha \fun \unit + \alpha$ for the embedding of $\alpha$ into $\unit + \alpha$ (i.e., $\embedding = \const{Inr}$).

%\begin{rep}
\begin{proof}
  From left to right is trivial with \theorem{map\_cong}.
  So let $x$ be such that $\const{map}_F\;\seq{f}\;x = \const{map}_F\;\seq{g}\;x$ whenever $f_i\;a = g_i\;a$ for all $a \in A_i$ and all $i$.
  By the definition of $\internalize{F}$, it suffices to show that $\const{set}_{F,i}\;x \subseteq A_i$.
  Set $f_i\;a = \embedding\;a$ if $a \in A_i$ and $f_i\;a = \circledast$ for $a \in A_i$, and $g_i = \embedding$.
  Then,
  \begin{equation*}
    \begin{tabular}{@{}l@{}l@{\qquad}l@{}}
      $f_i \langle \const{set}_{F,i}\;x \rangle$ &
      ${} = \const{set}_{F,i}\;(\const{map}_F\;\seq{f}\;x)$
      &
      by \theorem{set\_map}
      \\[0.7\jot]
      &
      ${} = \const{set}_{F,i}\;(\const{map}_F\;\seq{g}\;x)$
      &
      by choice of $x$ as $\seq{f}$ and $\seq{g}$ agree on $\seq{A}$
      \\[0.7\jot]
      &
      ${} = \embedding \langle \const{set}_{F,i}\;x\rangle$
      &
      by \theorem{set\_map}
    \end{tabular}
  \end{equation*}
  Therefore, $\forall a \in \const{set}_{F,i}\;x.\;\exists y.\; f_i\;a = \embedding\;y$, i.e., $\const{set}_{F,i}\;x \subseteq A_i$ by $f_i$'s definition.
\end{proof}
%\end{rep}
%
\indent
Equations~\ref{eq:set:in} and~\ref{eq:F_in} reduce the setters $\const{set}_{F,i}$ of a BNF to its mapper $\const{map}_F$.
In the next section, we will use this characterization to derive a definition of $\const{set}_{Q,i}$ in terms of $\const{set}_{F,i}$.
However, this definition does not give us naturality out of the box.

\begin{exaC}[{\cite[Example 4.2, part iii]{AdamekGummTrnkova2010JLC}}]%
  \label{ex:ae}
  Consider the functor $\alpha\;F_{\tc{seq}} = \tc{nat} \fun \alpha$ of infinite sequences with $x \sim_{\tc{ae}} y$ whenever $\{n \mid x\; n \neq y\; n\}$ is finite.
  That is, two sequences are equivalent iff they are equal almost everywhere.
  Conditions~\eqref{eq:equiv:sim} and~\eqref{eq:map:respect} hold, but not the naturality for the corresponding $\const{map}_Q$ and $\const{set}_Q$.
  \exampleend
\end{exaC}

Gumm~\cite{Gumm2005CALCO} showed that $\const{set}_F$ as defined in terms of~\eqref{eq:set:in} and~\eqref{eq:F_in} is a natural transformation iff $\internalize{F}$ preserves wide intersections and preimages, i.e.,
\begin{gather}
  \label{eq:wide:intersection}
  \internalize{F}\;\seq{(\bigcap \mathcal{A})} =
  \bigcap \{\internalize{F}\; \seq{A} \mid \forall i.\; A_i \in \mathcal{A}_i\}
  \\
  \label{eq:preimage:preservation}
  \internalize{F}\;\seq{(f^{-1}\langle A\rangle)} = (\const{map}_F\ \seq{f})^{-1}\langle\internalize{F}\;\seq{A}\rangle
\end{gather}
where $f^{-1}\langle A\rangle = \{ x \mid f\ x \in A \}$ denotes the preimage of $A$ under $f$.
Then, $\internalize{F}\;\seq{A} = \{ x \mid \forall i.\;\const{set}_{F,i}\;x \subseteq A_i\}$ holds.
The quotient in Example~\ref{ex:ae} does not preserve wide intersections.

In theory, we have now everything we need to define the BNF operations on the quotient $\seq{\alpha}\;Q = \seq{\alpha}\;F / \quotparen{\sim}$:
Define $\const{map}_Q$ as the lifting of $\const{map}_F$.
Define $\internalize{Q}$ and $\const{set}_{Q,i}$ using~\eqref{eq:F_in} and~\eqref{eq:set:in} in terms of $\const{map}_Q$, and the relator via \theorem{in\_rel}.
Prove that $\internalize{Q}$ preserves preimages and wide intersections.
Prove that $\const{rel}_Q$ satisfies subdistributivity (\theorem{rel\_comp}).

\looseness=-1
Unfortunately, the definitions and the preservation conditions are phrased in terms of $Q$, not in terms of $F$ and $\optparen{\sim}$.
It is therefore unclear how $\const{set}_{Q,i}$ and $\const{rel}_Q$ relate to $\const{set}_{F,i}$ and $\const{rel}_F$.
In practice, understanding this relationship is important:
we want to express the BNF operations and discharge the proof obligations in terms of $F$'s operations and later use the connection to transfer properties from $\const{set}_F$ and $\const{rel}_F$ to $\const{set}_Q$ and $\const{rel}_Q$.
We will work out the precise relationships for the setters in Section~\ref{section:repair} and for the relator in Section~\ref{section:relator}.

\subsection{The quotient's setter}%
\label{section:repair}

We relate $Q$'s setters to $F$'s operations and $\optparen{\sim}$.
We first look at $\internalize{Q}$, which characterizes $\const{set}_{Q,i}$ via~\eqref{eq:set:in}.
Let $[x]_\sim = \{y \mid x \sim y\}$ denote the equivalence class that $x \hastype \seq{\alpha}\;F$ belongs to, and $[A]_\sim = \{ [x]_\sim \mid x \in A\}$ denote the equivalence classes of elements in $A$.
We identify the values of $\seq{\alpha}\;Q$ with $\seq{\alpha}\;F$'s equivalence classes.
Then, it follows using~\eqref{eq:equiv:sim},~\eqref{eq:map:respect}, and~\eqref{eq:F_in} that $\internalize{Q}\;A = [\internalizesim{F}{\sim}\;A]_\sim$ where
\begin{equation}
  \label{eq:F_in'}
  \internalizesim{F}{\sim}\;\seq{A} = \{ x \mid \forall \seq{f}\;\seq{g}.\;(\forall i.\;\forall a\in A_i.\; f_i\; a = g_i\; a) \longrightarrow \const{map}_F\;\seq{f}\;x \sim \const{map}_F\;\seq{g}\;x\}
\end{equation}
with $f_i, g_i \hastype \alpha_i \fun \unit + \alpha_i$.
Equation~\ref{eq:F_in'} differs from~\eqref{eq:F_in} only in that the equality in $\const{map}_F\;\seq{f}\;x = \const{map}_F\;\seq{g}\;x$ is replaced by $\optparen{\sim}$.
Clearly $[\internalize{F}\;\seq{A}]_\sim \subseteq [\internalizesim{F}{\sim}\;\seq{A}]_\sim$.
The converse holds for non-empty sets $A_i$, as shown next.

\begin{lem}
  If $A_i \neq \{\}$ for all $i$,
  then $[\internalizesim{F}{\sim}\;\seq{A}]_\sim \subseteq [\internalize{F}\;\seq{A}]_\sim$.
\end{lem}

\begin{proof}
  Since $A_i$ is non-empty, fix $a_i \in A_i$ for all $i$.
  Let $x \in \internalizesim{F}{\sim}\;\seq{A}$ and consider $y = \const{map}_F\;\seq{h}\;x$
  where $h_i\;a = a$ if $a \in A_i$ and $h_i\;a = a_i$ otherwise.
  Let $\embedding^{-1}$ denote the left-inverse of $\embedding$.
  Then,
  \begin{equation*}
    \begin{tabular}{@{}l@{}l@{\qquad}l@{}}
      $x$ &
      ${} = \const{map}_F\;\seq{(\embedding^{-1} \circ \embedding)}\;x$
      &
      by \theorem{map\_id} and $\embedding^{-1} \circ \embedding = \const{id}$
      \\[0.7\jot]
      &
      ${} = \const{map}_F\;\seq{\embedding^{-1}}\;(\const{map}_F\;\seq{\embedding}\;x)$
      &
      by \theorem{map\_comp}
      \\[0.7\jot]
      &
      ${} \sim \const{map}_F\;\seq{\embedding^{-1}}\;(\const{map}_F\;\seq{(\embedding \circ h)}\;x)$
      &
      by~\eqref{eq:F_in'} and~\eqref{eq:map:respect} as $\embedding\;a = (\embedding \circ h_i)\;a$ for $a \in A_i$
      \\[0.7\jot]
      &
      ${} = \const{map}_F\;\seq{(\embedding^{-1} \circ \embedding \circ h)}\;x = y$
      &
      by \theorem{map\_comp} and $\embedding^{-1} \circ \embedding = \const{id}$
    \end{tabular}
  \end{equation*}
  It therefore suffices to show that $y \in \internalize{F}\;\seq{A}$.
  Let $\seq{f}$ and $\seq{g}$ with $f_i, g_i \hastype \alpha_i \fun \unit + \alpha_i$ such that $f_i\;a = g_i\;a$ for all $a \in A_i$.
  Then $f_i \circ h_i = g_i \circ h_i$ as the range of $h_i$ is $A_i$.
  So $\const{map}_F\;\seq{f}\;y = \const{map}_F\;\seq{(f \circ h)}\;x = \const{map}_F\;\seq{(g \circ h)}\;x = \const{map}_F\;\seq{g}\;y$.
  Thus $y \in \internalize{F}\;\seq{A}$.
\end{proof}

However, this inclusion $[\internalizesim{F}{\sim}\;\seq{A}]_\sim \subseteq [\internalize{F}\;\seq{A}]_\sim$ may fail for empty sets $A_i$, as the next example shows.

\begin{exa}[Example~\ref{ex:Inl:Inl} continued]%
  \label{ex:Inl:Inl:cont}
  For the example viewing $\unit + \alpha$ as a quotient of $\alpha\;F_P = {\alpha + \alpha}$ via $\optparen{\sim_P}$, we have $[\const{Inl}\;x]_{\sim_P} \in [\internalizesim{F_P}{\sim_P}\;\{\}]_{\sim_P}$ because $\const{map}_{F_P}\;f\;(\const{Inl}\;x) = \const{Inl}\;(f\;x) \sim_P \const{Inl}\;(g\;x) = \const{map}_{F_P}\;g\;(\const{Inl}\;x)$ for all $f$ and $g$.
  Yet $\internalize{F_P}\;\{\}$ is empty, and so is $[\internalize{F_P}\;\{\}]_{\sim_P}$.
  \exampleend
\end{exa}

To avoid the problematic case of empty sets, we change types:
instead of $\seq{\alpha}\;F / \quotparen{\sim}$, we consider the quotient $\seq{(\unit + \alpha)}\;F / \quotparen{\sim}$.
Then, we have the following equivalence: %, where $f\langle A\rangle = \{ f\;x \mid x \in A \}$ denotes the image of $A$ under $f$.

\begin{lem}%
  \label{lem:F_in'}
  $\internalizesim{F}{\sim}\;\seq{A} = \{ x \mid [\const{map}_F\;\seq{\embedding}\;x]_\sim \in [\internalize{F}\;\seq{(\{\circledast\} \cup \embedding\langle A\rangle)}]_\sim\}$.
\end{lem}
%\begin{conf}
%\pagebreak[2]
%\end{conf}
%
%\begin{rep}
\begin{proof}
  For the left to right direction,
  let $x \in \internalizesim{F}{\sim}\;\seq{A}$ and set $f_i\;y = \embedding\;y$ for $y \in A_i$ and $f_i\;y = \circledast$ for $y \notin A_i$.
  Then, $\const{set}_{F,i}\;(\const{map}_F\;\seq{f}\;x) = f_i\langle\const{set}_{F,i}\;x\rangle$ by the naturality of $\const{set}_{F,i}$ and $f_i\langle B\rangle \subseteq \{\circledast\} \cup \embedding\langle A_i\rangle$ by $f_i$'s definition for any $B$.
  Hence $\const{map}\;\seq{f}\;x \in \internalize{F}\;\seq{(\{\circledast\} \cup \embedding\langle A\rangle)}$ as $\internalize{F}\;\seq{C} = \{ x \mid \forall i.\;\const{set}_{F,i}\; x \subseteq C_i \}$ by definition.
  So, $[\const{map}_F\;\seq\embedding\;x]_\sim \in [\internalize{F}\;\seq{(\{\circledast\} \cup \embedding\langle A\rangle)}]_\sim$
  because $\const{map}_F\;\seq\embedding\;x \sim \const{map}\;\seq{f}\;x$ by~\eqref{eq:F_in'} and $x \in \internalizesim{F}{\sim}\;\seq{A}$.

  For the right to left direction,
  let $x$ such that $\const{map}_F\;\seq{\embedding}\;x \sim y$ for some $y \in \internalize{F}\;\seq{(\{\circledast\} \cup \embedding\langle A\rangle)}$.
  Let $\seq{f}$ and $\seq{g}$ such that $f_i\;a = g_i\;a$ for all $a \in A_i$ and all $i$.
  Then, $\const{map}_F\;\seq{f}\;x \sim \const{map}_F\;\seq{g}\;x$ holds by the following reasoning,
  where $\const{map}_{\unit+{}}\;h$ satisfies $\const{map}_{\unit+{}}\;h\;(\embedding\;a) = \embedding\;(h\;a)$ and $\const{map}_{\unit+{}}\;h\;\circledast = \circledast$:
  \begin{equation*}
    \begin{tabular}[b]{@{}l@{}l@{\qquad}l@{}}
      $\const{map}_F\;\seq{f}\;x$ &
      ${} =
      \const{map}_F\;\seq{\embedding^{-1}}\;(\const{map}_F\;\seq{(\const{map}_{\unit+{}}\;f)}\;(\const{map}_F\;\seq{\embedding}\;x))$
      &
      as $f_i = \embedding^{-1} \circ \const{map}_{\unit+{}}\;f_i \circ \embedding$
      \\[0.7\jot]
      & ${} \sim \const{map}_F\;\seq{\embedding^{-1}}\;(\const{map}_F\;\seq{(\const{map}_{\unit+{}}\;f)}\;y)$
      &
      by $\const{map}_F\;\seq{\embedding}\;x \sim y$ and~\eqref{eq:map:respect}
      \\[0.7\jot]
      & ${} = \const{map}_F\;\seq{\embedding^{-1}}\;(\const{map}_F\;\seq{(\const{map}_{\unit+{}}\;g)}\;y)$
      &
      by choice of $y$ and~\eqref{eq:F_in}
      \\[0.7\jot]
      & ${} \sim \const{map}_F\;\seq{\embedding^{-1}}\;(\const{map}_F\;\seq{(\const{map}_{\unit+{}}\;g)}\;(\const{map}_F\;\seq{\embedding}\;x))$
      &
      by $y \sim \const{map}_F\;\seq{\embedding}\;x$ and~\eqref{eq:map:respect}
      \\[0.7\jot]
      & ${} = \const{map}_F\;\seq{g}\;x$
      &
      as $\embedding^{-1} \circ \const{map}_{\unit+{}}\;g_i \circ \embedding = g_i$
    \end{tabular}
    \qedhere
  \end{equation*}
\end{proof}

Lemma~\ref{lem:F_in'} allows us to characterize the quotient's setters $\const{set}_{Q}$ in terms of $\const{set}_F$.

\begin{thm}[Setter characterization]%
  \label{thm:set:characterization}
  $\const{set}_{Q,i}\;[x]_\sim = \bigcap\nolimits_{y \in [\const{map}_F\;\seq{\embedding}\;x]_\sim} \{a \mid \embedding\;a \in \const{set}_{F,i}\;y\}$
\end{thm}
\begin{proof}
  Recall that we defined $\const{set}_{Q,i}$ by~\eqref{eq:set:in}.
  Then
  \begin{equation*}
    \begin{tabular}{@{}l@{}l@{}l@{\qquad\!\!\!}l@{}}
      $\const{set}_{Q,i}\ [x]_\sim$
      &
      ${} = {}$
      &
      $\bigcap \{ A_i \mid [x]_\sim \in \internalize{Q}\;\seq{\const{UNIV}}\;A_i\;\seq{\const{UNIV}} \}$
      &
      by~\eqref{eq:set:in}
      \\[0.7\jot]
      &
      ${} = {}$
      &
      $\bigcap \{ A_i \mid [x]_\sim \in [\internalizesim{F}{\sim}\;\seq{\const{UNIV}}\;A_i\;\seq{\const{UNIV}}]_\sim \}$
      &
      by $\internalize{Q}\;\seq{A} = [\internalizesim{F}{\sim}\;\seq{A}]_\sim$
      \\[0.7\jot]
      &
      ${} = {}$
      &
      $\bigcap \{ A_i \mid [\const{map}_F\;\seq{\embedding}\;x]_\sim \in [\internalize{F}\;\seq{\const{UNIV}}\;(\{\circledast\} \cup \embedding\langle A_i\rangle)\;\seq{\const{UNIV}}]_\sim \}$
      &
      by Lemma~\ref{lem:F_in'}
      \\[0.7\jot]
      &
      ${} = {}$
      &
      $\bigcap \{ A_i \mid [\const{map}_F\;\seq{\embedding}\;x]_\sim \in [\{ y \mid \const{set}_{F,i}\;y \subseteq \{\circledast\} \cup \embedding\langle A_i\rangle\}]_\sim \}$
      &
      by Definition of $\internalize{F}$
      \\[0.7\jot]
      &
      ${} = {}$
      &
      $\bigcap \{ \{ a \mid \embedding\;a \in \const{set}_{F,i}\; y \} \mid y \sim \const{map}_F\;\seq{\embedding}\;x \}$
      &
      \\[0.7\jot]
      &
      ${} = {}$
      &
      $\bigcap\nolimits_{y \in [\const{map}_F\;\seq{\embedding}\;x]_\sim} \{a \mid \embedding\;a \in \const{set}_{F,i}\;y\}$
    &\qedhere
    \end{tabular}
  \end{equation*}
\end{proof}

\begin{exa}[Example~\ref{ex:Inl:Inl:cont} continued]\label{ex:Inl:Inl:cont2}
  For the example viewing $\unit + \alpha$ as a quotient of $\alpha\;F_P = {\alpha + \alpha}$ via $\optparen{\sim_P}$, Theorem~\ref{thm:set:characterization} yields
  \begin{trivlist}
  \item
    \begin{minipage}[t]{0.5\textwidth}
      \centering
      \begin{tabular}{@{}l@{}l@{}}
        &
        $\const{set}_{Q_P}\;[\const{Inl}\;x]_{\sim_P}$
        \\[0.7\jot]
        ${}={}$
        &
        $\bigcap\nolimits_{y \in [\const{map}_{F_P}\;\embedding\;(\const{Inl}\;x)]_{\sim_P}} \{a \mid \embedding\;a \in \const{set}_{F_P}\;y\}$
        \\[0.7\jot]
        ${}={}$
        &
        $\bigcap\nolimits_{y \in [\const{Inl}\;(\embedding\;x)]_{\sim_P}} \{a \mid \embedding\;a \in \const{set}_{F_P}\;y\}$
        \\[0.7\jot]
        ${}={}$
        &
        $\bigcap\nolimits_{y \in \const{Inl}\left\langle\const{UNIV}\right\rangle} \{a \mid \embedding\;a \in \const{set}_{F_P}\;y\}$
        \\[0.7\jot]
        ${}={}$
        &
        $\bigcap\nolimits_{z \in \const{UNIV}} \{a \mid \embedding\;a \in \const{set}_{F_P}\;(\const{Inl}\;z)\}$
        \\[0.7\jot]
        ${}={}$
        &
        $\bigcap\nolimits_{z \in \const{UNIV}} \{a \mid \embedding\;a = z\}$
        \\[0.7\jot]
        ${}={}$
        &
        $\{\}$
      \end{tabular}
    \end{minipage}%
\vrule
    \begin{minipage}[t]{0.5\textwidth}
      \centering
      \begin{tabular}{@{}l@{}l@{}}
        &
        $\const{set}_{Q_P}\;[\const{Inr}\;x]_{\sim_P}$
        \\[0.7\jot]
        ${}={}$
        &
        $\bigcap\nolimits_{y \in [\const{map}_{F_P}\;\embedding\;(\const{Inr}\;x)]_{\sim_P}} \{a \mid \embedding\;a \in \const{set}_{F_P}\;y\}$
        \\[0.7\jot]
        ${}={}$
        &
        $\bigcap\nolimits_{y \in [\const{Inr}\;(\embedding\;x)]_{\sim_P}} \{a \mid \embedding\;a \in \const{set}_{F_P}\;y\}$
        \\[0.7\jot]
        ${}={}$
        &
        $\bigcap\nolimits_{y \in \{\const{Inr}\;(\embedding\;x)\}} \{a \mid \embedding\;a \in \const{set}_{F_P}\;y\}$
        \\[0.7\jot]
        ${}={}$
        &
        $\{a \mid \embedding\;a \in \const{set}_{F_P}\;(\const{Inr}\;(\embedding\;x))\}$
        \\[0.7\jot]
        ${}={}$
        &
        $\{a \mid \embedding\;a = \embedding\;x\}$
        \\[0.7\jot]
        ${}={}$
        &
        $\{x\}$
      \end{tabular}
    \end{minipage}
  \end{trivlist}
  \vspace{-0.8em}\exampleend
\end{exa}

Next, we express the conditions~\eqref{eq:wide:intersection} and~\eqref{eq:preimage:preservation} on $\internalize{Q}$ in terms of $\optparen{\sim}$ and~$\internalize{F}$.
For wide intersections, the condition is as follows:
%\begin{conf}%
%\pagebreak[4]%
%\end{conf}
\begin{equation}
  \label{eq:wide:intersections}
  \kern-.5em
  \forall i.\;\mathcal{A}_i \neq \{\} \wedge (\bigcap\mathcal{A}_i \neq \{\})
  {\implies}
  \bigcap\{[\internalize{F}\;\seq{A}]_\sim \mid \forall i.\;A_i \in \mathcal{A}_i\}  \subseteq
  \left[\bigcap\{\internalize{F}\;\seq{A} \mid \forall i.\;A_i \in \mathcal{A}_i\}\right]_\sim\kern-.9em
\end{equation}
\looseness=-1
The conclusion is as expected: for sets of the form $\internalize{F}\;\seq{A}$, taking equivalence classes preserves wide intersections.
The assumption is the interesting part:
preservation is needed only for \emph{non-empty} intersections.
Non-emptiness suffices because Lemma~\ref{lem:F_in'} relates $\internalizesim{F}{\sim}\;\seq{A}$ to $\internalize{F}\;\seq{(\{\circledast\} \cup \embedding\langle A\rangle)}$ and all intersections of interest therefore contain $\circledast$. %
%\footnote{%
%  (The condition does not explicitly mention $\circledast$ because Lemma~\ref{lem:F_in'} holds for any element that is not in $A$.)
%}

\begin{lem}%
  \label{lemma:wide:intersection:preservation}
  $[\internalizesim{F}{\sim}\;\seq{(\bigcap \mathcal{B})}]_\sim = \bigcap \{[\internalizesim{F}{\sim}\; \seq{B}]_\sim \mid \forall i.\; B_i \in \mathcal{B}_i\}$
  if~\eqref{eq:wide:intersections} holds for $\seq{\mathcal{A}}$ given by $\mathcal{A}_i = \{ \{\circledast\} \cup \embedding\langle B\rangle \mid B \in \mathcal{B}_i \}$.
\end{lem}

\begin{proof}
  Note that $\seq{\mathcal{A}}$ satisfies the assumption of~\eqref{eq:wide:intersections}.
  We first show that the other inclusion of~\eqref{eq:wide:intersections} holds trivially.
  Let $u \sim x \in \bigcap\{\internalize{F}\;\seq{A} \mid \forall i.\;A_i \in \mathcal{A}_i\}$.
  Then $x \in \internalize{F}\;\seq{A}$ whenever $A_i \in \mathcal{A}_i$ for all $i$, and so is $u \in [\internalize{F}\;\seq{A}]_\sim$.
  Hence $\bigcap\{[\internalize{F}\;\seq{A}]_\sim \mid \forall i.\;A_i \in \mathcal{A}_i\} =
  \left[\bigcap\{\internalize{F}\;\seq{A} \mid \forall i.\;A_i \in \mathcal{A}_i\}\right]_\sim$.

  As $\embedding$ is injective and $\circledast$ is not in $\embedding$'s range, we have $\bigcap \mathcal{A}_i = \{\circledast\} \cup \embedding\langle\bigcap \mathcal{B}_i\rangle$.
  We calculate
  \begin{equation*}
    \begin{tabular}{@{}l@{}l@{}l@{\qquad}l@{}}
      $[\internalizesim{F}{\sim}\;\seq{(\bigcap \mathcal{B})}]_\sim$
      &
      ${} = {}$
      &
      $[\{ x \mid [\const{map}_F\;\seq{\embedding}\;x]_\sim \in [\internalize{F}\;\seq{(\{\circledast\} \cup \embedding\langle \bigcap \mathcal{B}\rangle)}]_\sim \}]_\sim$
      &
      by Lemma~\ref{lem:F_in'}
      \\[0.7\jot]
      &
      ${} = {}$
      &
      $[\{ x \mid [\const{map}_F\;\seq{\embedding}\;x]_\sim \in [\internalize{F}\;\seq{(\bigcap\mathcal{A})}]_\sim \}]_\sim$
      \\[0.7\jot]
      &
      ${} = {}$
      &
      $[\{ x \mid [\const{map}_F\;\seq{\embedding}\;x]_\sim \in \left[ \bigcap \{ \internalize{F}\;\seq{A}  \mid \forall i.\; A_i \in \mathcal{A}_i \}\right]_\sim \}]_\sim$
      &
      by~\eqref{eq:wide:intersection}
      \\[0.7\jot]
      &
      ${} = {}$
      &
      $[\{ x \mid [\const{map}_F\;\seq{\embedding}\;x]_\sim \in \bigcap \{ [\internalize{F}\;\seq{A}]_\sim \mid \forall i.\; A_i \in \mathcal{A}_i \}]_\sim$
      &
      by the above equality
      \\[0.7\jot]
      &
      ${} = {}$
      &
      $\bigcap \left\{ [\{ x \mid [\const{map}_F\;\seq{\embedding}\;x]_\sim \in [\internalize{F}\;\seq{A}]_\sim \}]_\sim \mid \forall i.\; A_i \in \mathcal{A}_i \right\}$
      &
      \\[0.7\jot]
      &
      ${} = {}$
      &
      $\bigcap \{ [\internalizesim{F}{\sim}\;\seq{B}]_\sim \mid \forall i.\; B_i \in \mathcal{B}_i \}$
      &
      by Lemma~\ref{lem:F_in'}
    \qedhere
    \end{tabular}
  \end{equation*}
\end{proof}

Condition~\ref{eq:wide:intersections} is satisfied trivially for equivalence relations that preserve $\const{set}_{F,i}$, i.e., satisfy~\eqref{eq:set:respect:naive}.
Examples include permutative structures like finite sets and cyclic lists.

\begin{lem}%
  \label{lem:wide:intersections}
  If $\optparen{\sim}$ satisfies~\eqref{eq:set:respect:naive},
  then $[x]_\sim \in [\internalize{F}\;\seq{A}]_\sim$ iff $x \in \internalize{F}\;\seq{A}$, and
  condition~\eqref{eq:wide:intersections} holds.
\end{lem}
%
%\begin{rep}
\begin{proof}
  Condition~\eqref{eq:set:respect:naive} says that $\const{set}_{F,i}\;x = \const{set}_{F,i}\;y$ whenever $x \sim y$.
  So $[x]_\sim \in [\internalize{F}\;\seq{A}]_\sim$ iff $x \in \internalize{F}\;\seq{A}$ because $\internalize{F}\;\seq{A} = \{ x \mid \forall i.\;\const{set}_{F,i}\;x \subseteq A_i\}$ for all $\seq{A}$.
  Thus,~\eqref{eq:wide:intersections} holds by the following calculation:
  \begin{equation*}
    \begin{tabular}{@{}l@{\ iff\ }l@{\qquad}l@{}}
      $[x]_\sim \in \bigcap\{[\internalize{F}\;\seq{A}]_\sim \mid \forall i.\;A_i \in \mathcal{A}_i\}$
      &
      $[x]_\sim \in [\internalize{F}\;\seq{A}]_\sim$ whenever $A_i \in \mathcal{A}_i$ for all $i$
      \\[0.7\jot]
      &
      $x \in \internalize{F}\;\seq{A}$ whenever $A_i \in \mathcal{A}_i$ for all $i$
      \\[0.7\jot]
      &
      $x \in \internalize{F}\;\seq{(\bigcap \mathcal{A})}$
      &
      by~\eqref{eq:wide:intersection}
      \\[0.7\jot]
      &
      $[x]_\sim \in [\internalize{F}\;\seq{(\bigcap \mathcal{A})}]_\sim$
      &
    \qedhere
    \end{tabular}%
  \end{equation*}
\end{proof}
%\end{rep}
%
In contrast, the non-emptiness assumption is crucial for quotients that identify values with different sets of atoms, such as Example~\ref{ex:Inl:Inl}.
In general, such quotients do \emph{not} preserve empty intersections (Section~\ref{section:related:work}).

We can factor condition~\eqref{eq:wide:intersections} into a separate property for each type argument $i$: %\pagebreak[2]
\begin{equation}
  \label{eq:wide:intersection:one}
  \kern-.5em
  \mathcal{A}_i \neq \{\} \wedge (\bigcap\mathcal{A}_i) \neq \{\}
  {\implies}
  \bigcap\nolimits_{A \in \mathcal{A}_i} [\{x \mid \const{set}_{F,i}\;x \subseteq A\}]_\sim \subseteq
  \left[\{x \mid \const{set}_{F,i}\;x \subseteq \bigcap \mathcal{A}_i\}\right]_\sim\kern-1em
\end{equation}
This form is used in our implementation (Section~\ref{section:implementation}).
It is arguably more natural to prove for a concrete functor $F$ because each property focuses on a single setter.
\begin{lem}%
  \label{lem:wide:intersections:alt}
  Let $\optparen{\sim}$ satisfy~\eqref{eq:equiv:sim} and~\eqref{eq:map:respect}.
  Then,~\eqref{eq:wide:intersections} holds iff~\eqref{eq:wide:intersection:one} holds for all $i$.
 %condition
 %wide intersections are preserved condition
\end{lem}

%\begin{rep}
\begin{proof}
  $\eqref{eq:wide:intersections} \implies~\eqref{eq:wide:intersection:one}$ follows directly by setting $\mathcal{A}_j = \{\const{UNIV}\}$ for all $j \neq i$, where $\mathcal{\const{UNIV}}$ is the universe of the respective type.
  For the other direction, fix $x$ such that for all $\seq{A}$ where $\forall i.\;A_i \in \mathcal{A}_i$, there exists $y_{\seq{A}} \in \internalize{F}\;\seq{A}$ such that $x \sim y_{\seq{A}}$.
  For every $i$, we have $\const{set}_{F,i}\;y_{\seq{A}} \subseteq B$ for $B \in \mathcal{A}_i$ and hence $x \in \bigcap\nolimits_{B \in \mathcal{A}_i} [\{x \mid \const{set}_{F,i}\;x \subseteq B\}]_\sim$.
  By~\eqref{eq:wide:intersection:one} there exists $y_i$ such that $\const{set}_{F,i}\;y_i \subseteq \bigcap \mathcal{A}_i$ and $x \sim y_i$.
  Fix an arbitrary $a_i \in \bigcap \mathcal{A}_i$, which exists because $\bigcap \mathcal{A}_i$ is assumed to be non-empty.
  Moreover, set
  \begin{equation*}
    f_i\;x = \begin{cases}
      x &\text{if $x \in \bigcap \mathcal{A}_i$}\\
      a_i &\text{otherwise}
    \end{cases}
  \end{equation*}
  We calculate
  \begin{equation*}
    \begin{tabular}{@{}l@{}l@{\qquad}l@{}}
      $x$ & ${} \sim \const{map}_F\;\seq{\const{id}}\;y_1$ & by $x \sim y_1$ and \theorem{map\_id}\\[0.7\jot]
      &${} = \const{map}_F\;f_1\;\seq{\const{id}}\;y_1$ & by $\const{set}_{F,1}\;y_1 \subseteq \bigcap \mathcal{A}_1$ and \theorem{map\_cong}\\[0.7\jot]
      &${}\sim \const{map}_F\;f_1\;\seq{\const{id}}\;y_2$ & by~\eqref{eq:map:respect} and $y_1 \sim x \sim y_2$\\[0.7\jot]
      &${} = \const{map}_F\;f_1\;f_2\;\seq{\const{id}}\;y_2$ & by $\const{set}_{F,2}\;y_2 \subseteq \bigcap \mathcal{A}_2$ and \theorem{map\_cong}\\[0.7\jot]
      &${} \sim \dots \sim \const{map}_F\;\seq{f}\;y_n$ & similarly \\[0.7\jot]
      &${} \sim \const{map}_F\;\seq{f}\;x$ & by $x \sim y_n$ and \theorem{map\_cong}
    \end{tabular}
  \end{equation*}
  Finally, observe that $\const{map}_F\;\seq{f}\;x \in \internalize{F}\;\seq{A}$ whenever $\forall i.\;A_i \in \mathcal{A}_i$:
  We have $\const{set}_{F,i}\;(\const{map}_F\;\seq{f}\;x) = f_i\langle \const{set}_{F,i}\;x\rangle \subseteq \bigcap\mathcal{A}_i \subseteq A_i$ using \theorem{set\_map} and the definition of $f_i$.
\end{proof}
%\end{rep}

Many functors in practice contain only finitely many elements, i.e., $\const{set}_{F,i}\; x$ is always finite.
This includes all inductive datatypes built only from sums and products, e.g., finite lists and finitely branching trees.
Condition~\eqref{eq:wide:intersection:one} is always satisfied for such functors,
because wide intersections boil down to finite intersections in this case and Trnkov\'a~\cite{Trnkova1969CMUC} showed that all Set functors preserve non-empty binary intersections.

\begin{lem}%
  \label{lemma:wide:intersection:finite}
  If $\const{set}_{F,i}\; x$ is finite for all $x$, then~\eqref{eq:wide:intersection:one} holds for all equivalence relations $\sim$ that satisfy~\eqref{eq:map:respect}.
\end{lem}

\begin{proof}
  Fix $a_0 \in \bigcap \mathcal{A}_i$ and let $x \in \bigcap_{A \in \mathcal{A}_i} [\{x \mid \const{set}_{F,i}\;x \subseteq A\}]_\sim$.
  So for $A \in \mathcal{A}_i$, there exists a $y_A$ such that $x \sim y_A$ and $\const{set}_{F,i}\;y_A\subseteq A$.
  Set $B_A = \const{set}_{F,i}\;y_A \cup \{a_0\}$ and $\mathcal{B} = \{ B_A \mid A \in \mathcal{A}_i \}$.
  Then $\bigcap \mathcal{B} \subseteq \bigcap \mathcal{A}_i$.
  Since all $B_A$ are finite, there exist finitely many $A_j$ ($j=1,\ldots,n$ for some $n$) such that $\bigcap_{j=1}^n B_{A_j} \subseteq \bigcap \mathcal{B}$;
  for example, pick $A_1 \in \mathcal{A}_i$ arbitrarily, let $B_{A_1} - \bigcap \mathcal{B} = \{ b_2, \ldots, b_n \}$, and choose $A_j \in \mathcal{A}_i$ such that $b_j \notin B_{A_j}$ for $j = 2,\ldots,n$.
  Clearly, $a_0 \in B_{A_j}$ for all $j$ by construction.

  It suffices to show that there exists a $z$ such that $x \sim z$ and $\const{set}_{F,i}\;z \subseteq \bigcap_{j=1}^n B_{A_j}$.
  This $z$ proves that $x \in \left[\{x \mid \const{set}_{F,i}\;x \subseteq \bigcap \mathcal{A}_i\}\right]_\sim$
  because $\bigcap_{j=1}^n B_{A_j} \subseteq \bigcap \mathcal{B} \subseteq \bigcap \mathcal{A}_i$.
  We show the existence of such a $z$ by induction over the finitely many $A_j$.
  If there is only a single $A_1$, i.e., $n = 1$, then choose $z = y_{A_1}$.
  Otherwise, let $z'$ be such that $x \sim z'$ and $\const{set}_{F,i}\;z' \subseteq \bigcap_{j=1}^{n-1} B_{A_j}$ by the induction hypothesis.
  Define $f_i(a) = a$ for $a \in B_{A_n}$ and $f_i(a) = a_0$ for $a \notin B_{A_n}$.
  Set $f_{i'} = \const{id}$ for $i' \neq i$ and choose $z = \const{map}_F\;\seq{f}\;z'$.
  Then, $z \sim \const{map}_F\;\seq{f}\;y_{A_n}$ follows from $z' \sim x \sim y_{A_n}$ by~\eqref{eq:map:respect} and $\const{map}_F\;\seq{f}\;y_{A_n} = \const{map}_F\;\seq{\const{id}}\;y_{A_n} = y_{A_n}$ holds by \theorem{map\_cong} and \theorem{map\_id}, so $z \sim y_{A_n} \sim x$.
  Moreover, \theorem{set\_map} gives $\const{set}_{F,i}\;z = f_i\langle \const{set}_{F,i}\;z'\rangle \subseteq f_i\langle \bigcap_{j=1}^{n-1} B_{A_j}\rangle \subseteq \left(\bigcap_{j=1}^{n-1} B_{A_j}\right) \cap B_{A_n} = \bigcap_{j=1}^n B_{A_j}$.
\end{proof}

Preservation of preimages amounts to the following unsurprising condition:
\begin{equation}
  \label{eq:preimage:preservation:alt}
  \forall i.\;f_i^{-1}\langle A_i\rangle \neq \{\}
  \implies
  (\const{map}_F\;\seq{f})^{-1}\left\langle\bigcup[\internalize{F}\;\seq{A}]_\sim\right\rangle \subseteq
  \bigcup\left[(\const{map}_F\;\seq{f})^{-1}\langle\internalize{F}\;\seq{A}\rangle\right]_\sim
\end{equation}
As for wide intersections, taking equivalence classes must preserve \emph{non-empty} preimages.
Again, non-emptiness comes from $\circledast$ being contained in all sets of interest.

\begin{lem}%
  \label{lem:preimage:preservation}
  $\internalizesim{F}{\sim}\ \seq{(g^{-1}\langle B\rangle)} = (\const{map}_F\;\seq{g})^{-1} \langle\internalizesim{F}{\sim}\;\seq{B}\rangle$
  if~\eqref{eq:preimage:preservation:alt} is satisfied for $\seq{A}$ and $\seq{f}$ given by $A_i = \{\circledast\} \cup \embedding\langle B_i \rangle$ and $f_i = \const{map}_{1+}\;g_i$.
\end{lem}
\begin{proof}
  We have $f_i^{-1}\langle A_i \rangle = \{\circledast\} \cup \embedding\langle g_i^{-1}\langle B_i\rangle\rangle$ and the precondition of~\eqref{eq:preimage:preservation:alt} holds.
  The inclusion from right to left holds trivially:
  If $x \in \bigcup\left[(\const{map}_F\;\seq{f})^{-1}\langle\internalize{F}\;\seq{A}\rangle\right]_\sim$,
  then there exists a $y$ such that $x \sim y$ and $\const{map}_F\;\seq{f}\;y \in \internalize{F}\;A$;
  so $\const{map}_F\;\seq{f}\;x \sim \const{map}_F\;\seq{f}\;y$ by~\eqref{eq:map:respect}.
  We calculate
  \begin{equation*}
    \begin{tabular}[b]{@{}l@{}l@{}l@{\qquad}l@{}}
      $x \in \internalizesim{F}{\sim}\ \seq{(g^{-1}\langle B\rangle)}$
      &
      ${} \Iff {}$
      &
      $[\const{map}_F\;\seq{\embedding}\;x]_\sim \in [\internalize{F}\;\seq{(\{\circledast\} \cup \embedding \langle g^{-1}\langle B \rangle\rangle)}]_\sim \}$
      &
      by Lemma~\ref{lem:F_in'}
      \\[0.7\jot]
      &
      ${} \Iff {}$
      &
      $[\const{map}_F\;\seq{\embedding}\;x]_\sim \in [\internalize{F}\;\seq{(f^{-1}\langle A \rangle)}]_\sim$
      &
      by choice of $\seq{f}$
      \\[0.7\jot]
      &
      ${} \Iff {}$
      &
      $[\const{map}_F\;\seq{\embedding}\;x]_\sim \in [(\const{map}_F\;\seq{f})^{-1}\langle \internalize{F}\;\seq{A}\rangle]_\sim$
      &
      by~\eqref{eq:preimage:preservation}
      \\[0.7\jot]
      &
      ${} \Iff {}$
      &
      $\const{map}_F\;\seq{\embedding}\;x \in \bigcup[(\const{map}_F\;\seq{f})^{-1}\langle \internalize{F}\;\seq{A}\rangle]_\sim$
      \\[0.7\jot]
      &
      ${} \Iff {}$
      &
      $\const{map}_F\;\seq{\embedding}\;x \in (\const{map}_F\;\seq{f})^{-1}\langle \bigcup [\internalize{F}\;\seq{A}]_\sim \rangle$
      &
      by~\eqref{eq:preimage:preservation:alt}
      \\[0.7\jot]
      &
      ${} \Iff {}$
      &
      $[\const{map}_F\;\seq{f}\;(\const{map}_F\;\seq{\embedding}\;x)]_\sim \in [\internalize{F}\;\seq{A}]_\sim$
      \\[0.7\jot]
      &
      ${} \Iff {}$
      &
      $[\const{map}_F\;\seq{\embedding}\;(\const{map}_F\;\seq{g}\;x)]_\sim \in [\internalize{F}\;\seq{A}]_\sim$
      &
      as $f_i \circ \embedding = \embedding \circ g_i$
      \\[0.7\jot]
      &
      ${} \Iff {}$
      &
      $x \in (\const{map}_F\;\seq{g})^{-1} \langle\internalizesim{F}{\sim}\;\seq{B}\rangle$
      &
      by Lemma~\ref{lem:F_in'}
    \end{tabular}
    \qedhere
  \end{equation*}
\end{proof}

We do not elaborate on how to establish preimage preservation~\eqref{eq:preimage:preservation:alt} any further as it follows from subdistributivity, which we will look at in the next subsection.
Instead, we show that the quotient setters $\const{set}_Q$ are natural transformations under conditions~\eqref{eq:wide:intersections} and~\eqref{eq:preimage:preservation:alt}.

\begin{lem}%
  \label{lemma:set:Q:natural}
  Under the conditions~\eqref{eq:wide:intersections} and~\eqref{eq:preimage:preservation:alt},
  $\const{set}_{Q,i}$ is a natural transformation, i.e.,
  $\const{set}_{Q,i}\;(\const{map}_Q\;\seq{f}\;[x]_\sim) = f_i\langle \const{set}_{Q,i}\;[x]_\sim\rangle$,
  and therefore $\const{set}_{Q,i}\;[\const{map}_F\;\seq{f}\;x]_\sim = f_i\langle \const{set}_{Q,i}\;[x]_\sim\rangle$.
\end{lem}
\begin{proof}
  By Lemma~\ref{lemma:wide:intersection:preservation}, $\internalize{Q}$ preserves wide intersections~\eqref{eq:wide:intersection}.
  Similarly, Lemma~\ref{lem:preimage:preservation} shows that $\internalize{Q}$ preserves preimages~\eqref{eq:preimage:preservation}.
  The claim follows with~\cite[Theorem~6]{Gumm2005CALCO}.
\end{proof}

\begin{lem}%
  \label{lemma:set:Q:map:cong}
  Under the conditions~\eqref{eq:wide:intersections} and~\eqref{eq:preimage:preservation:alt},
  $\const{map}_F\;\seq{f}\;x \sim \const{map}_F\;\seq{g}\;x$
  if $f_i\;z = g_i\;z$ for all $i$ and all $z \in \const{set}_{Q,i}\;[x]_\sim$.
\end{lem}
\begin{proof}
  We first show that $[x]_\sim \in \internalize{Q}\;\seq{(\const{set}_Q\;[x]_\sim)}$.
  We observe
  \begin{equation*}
    \begin{tabular}{@{}l@{}l@{}l@{\quad}l@{}}
      $[x]_\sim \in \internalize{Q}\;\seq{(\const{set}_Q\;[x]_\sim)}$
      &
      ${} \Iff {}$
      &
      $[x]_\sim \in \internalize{Q}\;\seq{(\bigcap_{y \in [\const{map}\;\seq\embedding\;x]_\sim}\;\{ a \mid \embedding\;a \in \const{set}_F\;y \})}$
      &
      by Theorem~\ref{thm:set:characterization}
      \\[0.7\jot]
      &
      ${} \Iff {}$
      &
      $[x]_\sim \in \bigcap_{y \in [\const{map}\;\seq{\embedding}\;x]_\sim} [\internalizesim{F}{\sim}\;\seq{\{ a \mid \embedding\;a \in \const{set}_F\;y \}}]_\sim$
      &
      by Lemma~\ref{lemma:wide:intersection:preservation}
      \\[0.7\jot]
      &
      ${} \Iff {}$
      &
      $\forall y \in [\const{map}\;\seq{\embedding}\;x]_\sim.\;
      [x]_\sim \in [\internalizesim{F}{\sim}\;\seq{\{ a \mid \embedding\;a \in \const{set}_F\;y \}}]_\sim$
    \end{tabular}
  \end{equation*}
  So let $y \sim \const{map}\;\seq{\embedding}\;x$.
  Then $y \in \internalize{F}\;\seq{(\{\circledast\} \cup \const{set}_F\;y)}$ by definition of $\internalize{F}$ and thus $[x]_\sim \in [\internalizesim{F}{\sim}\;\seq{\{ a \mid \embedding\;a \in \const{set}_F\;y \}}]_\sim$ by Lemma~\ref{lem:F_in'} using $\embedding\langle \{ a \mid \embedding\;a \in \const{set}_{F,i}\;y \} \rangle = \const{set}_{F,i}\;y$ for all $i$.

  By~\eqref{eq:F_in'}, we therefore have $\const{map}_F\;\seq{\embedding}\;x \sim \const{map}_F\;\seq{h}\;x$ for $h_i\;a = \embedding\;a$ for $a \in \const{set}_{Q,i}$ and $h_i\;a = \circledast$ otherwise.
  Then
  \begin{equation*}
    \begin{tabular}[b]{@{}l@{}l@{\qquad}l@{}}
      $\const{map}_F\;\seq{f}\;x$
      &
      ${} = \const{map}_F\;\seq{\embedding^{-1}}\;(\const{map}_F\;\seq{(\const{map}_{\unit+}\;f)}\;(\const{map}_F\;\seq{\embedding}\;x))$
      &
      as $\embedding^{-1} \circ \const{map}_{\unit+}\;f_i \circ \embedding = f_i$
      \\[0.7\jot]
      &
      ${} \sim \const{map}_F\;\seq{\embedding^{-1}}\;(\const{map}_F\;\seq{(\const{map}_{\unit+}\;f)}\;(\const{map}_F\;\seq{h}\;x))$
      &
      by~\eqref{eq:map:respect}
      \\[0.7\jot]
      &
      ${} = \const{map}_F\;\seq{\embedding^{-1}}\;(\const{map}_F\;\seq{(\const{map}_{\unit+}\;g)}\;(\const{map}_F\;\seq{h}\;x))$
      &
      by \theorem{map\_cong}
      \\[0.7\jot]
      &
      ${} \sim \const{map}_F\;\seq{\embedding^{-1}}\;(\const{map}_F\;\seq{(\const{map}_{\unit+}\;g)}\;(\const{map}_F\;\embedding\;x))$
      &
      by~\eqref{eq:map:respect}
      \\[0.7\jot]
      &
      ${} = \const{map}_F\;\seq{g}\;x$
      &
      as $\embedding^{-1} \circ \const{map}_{\unit+}\;g_i \circ \embedding = g_i$
    \end{tabular}
    \qedhere
  \end{equation*}
\end{proof}

\subsection{The quotient's relator}%
\label{section:relator}

In the previous section, we have shown that it is not a good idea to naively lift the setter and a more general construction is needed.
We now show that the same holds for the relator.
The following straightforward definition
\begin{quotex}
  $\keyw{lift\_definition}\;\const{rel}_Q \hastype
  \seq{(\alpha \reltype \beta)} \fun \seq{\alpha}\;Q \reltype \seq{\beta}\;Q\;
  \keyw{is}\;
  \const{rel}_F$
\end{quotex}
relates two equivalence classes $[x]_\sim$ and $[y]_\sim$ iff there are representatives $x' \in [x]_\sim$ and $y' \in [y]_\sim$ such that $(x', y') \in \const{rel}_F\;\seq{R}$. %
%DT redundant
%\footnote{%
%  In this paper, we model binary relations as sets of pairs, whereas they are modeled as binary predicates in Isabelle.
%  The corresponding definition in Isabelle for $\const{rel}_Q$ therefore lifts the term $\lambda \seq{R}.\; \optparen{\sim} \relcomp \const{rel}_F\;\seq{R} \relcomp \optparen{\sim}$ instead of $\const{rel}_F$.
%}
This relator does not satisfy \theorem{in\_rel} in general.
% the BNF relationship with $\const{map}_Q$ and $\internalize{F}$:
%\fxnote{Move characterization to background section?}
%\begin{equation}
%  \label{eq:rel_F}
%  (x, y) \in \const{rel}_Q\;\seq{R} \Iff
%  \exists z \in \internalize{Q}\;\seq{R}.\; x = \const{map}_Q\;\seq{\fst}\;z \wedge y = \const{map}_Q\;\seq{\snd}\;z
%\end{equation}

\begin{exa}[Example~\ref{ex:Inl:Inl:cont2} continued]\label{ex:Inl:Inl:cont3}
  By the lifted definition, $([\const{Inl}\;x]_{\sim_P},[ \const{Inl}\;y]_{\sim_P}) \notin \const{rel}_{Q_P}\;\{\}$ because there are no $(x', y')$ in the empty relation $\{\}$ that could be used to relate using $\const{rel}_{F_P}$ the representatives $\const{Inl}\;x'$ and $\const{Inl}\;y'$.
  However, the witness $z = [\const{Inl}\;(x, y)]_{\sim_P}$ satisfies the right-hand side of \theorem{in\_rel} as $\internalize{Q}\;\{\} = \{[\const{Inl}\;\_]_{\sim_P}\}$.
  \exampleend
\end{exa}

So what is the relationship between $\const{rel}_Q$ and $\const{rel}_F$ and under what conditions does the subdistributivity property \theorem{rel\_comp} hold?
Like for the setter, we avoid the problematic case of empty relations by switching to $\unit + \alpha$.
The relator $\const{rel}_{\unit+{}}$ adds the pair $(\circledast, \circledast)$ to every relation $R$ and thereby ensures that all relations and their compositions are non-empty.
Accordingly, we obtain the following characterization under the conditions~\eqref{eq:wide:intersections} and~\eqref{eq:preimage:preservation:alt}:

\begin{thm}[Relator characterization]%
  \label{thm:rel:characterization}
  \begin{equation*}
  ([x]_\sim, [y]_\sim) \in \const{rel}_Q\;\seq{R} \Iff
  (\const{map}_F\;\seq{\embedding}\;x, \const{map}_F\;\seq{\embedding}\;y) \in (\optparen{\sim} \relcomp \const{rel}_F\;\seq{(\const{rel}_{\unit+{}}\;R)} \relcomp \optparen{\sim})
  \end{equation*}
\end{thm}
\begin{proof}
  Applying \theorem{in\_rel} to both $\const{rel}_Q$ and $\const{rel}_F$, it suffices to show that
  \begin{equation}\label{eq:relQ:unfolded}
    \exists z.\; (\forall i.\; \const{set}_{Q,i}\;[z]_\sim \subseteq R_i) \land \const{map}_F\;\seq{\fst}\;z \sim x \land \const{map}_F\;\seq{\snd}\;z \sim y
  \end{equation}
  if and only if
  \begin{equation}\label{eq:relQ:alt:unfolded}
    \exists z'.\; (\forall i.\; \const{set}_{F,i}\;z' \subseteq \const{rel}_{\unit+{}}\;R_i) \land \const{map}_F\;\seq{\fst}\;z' \sim \const{map}_F\;\seq{\embedding}\;x \land \const{map}_F\;\seq{\snd}\;z' \sim \const{map}_F\;\seq{\embedding}\;y.
  \end{equation}
  Specifically, we show how to convert the witnesses $z$ and $z'$.

  From~\eqref{eq:relQ:unfolded} to~\eqref{eq:relQ:alt:unfolded}:
  Let the function $f_i$ send the pair $(a,b)$ to $(\embedding\;a,\, \embedding\;b)$ if $(a,b) \in R_i$ and otherwise to $(\circledast,\circledast)$.
  We choose $z' = \const{map}_F\;\seq{f}\;z$ and prove that it is a witness for~\eqref{eq:relQ:alt:unfolded}.
  By \theorem{set\_map} we have $\const{set}_{F,i}\;z' = f_i\langle \const{set}_{F,i}\;z\rangle$.
  The image of any set under $f_i$ is clearly included in $\const{rel}_{\unit+{}}\;R_i$, hence $\const{set}_{F,i}\;z' \subseteq \const{rel}_{\unit+{}}\;R_i$ for all $i$.
  Next, we calculate
  \begin{equation*}
    \begin{tabular}{@{}l@{}l@{\qquad}l@{}}
      $\const{map}_F\;\seq{\fst}\;z'$ & ${} = \const{map}_F\;\seq{(\fst \circ f)}\;z$ & by \theorem{map\_comp} \\[0.7\jot]
      & ${} \sim \const{map}_F\;\seq{(\embedding \circ \fst)}\;z$ & by Lemma~\ref{lemma:set:Q:map:cong} and $\forall i.\;\const{set}_{Q,i}\;[z]_\sim \subseteq R_i$ \\[0.7\jot]
      & ${} = \const{map}_F\;\seq{\embedding}\;(\const{map}_F\;\seq{\fst}\;z)$ & by \theorem{map\_comp} \\[0.7\jot]
      & ${} \sim \const{map}_F\;\seq{\embedding}\;x$ & by~\eqref{eq:map:respect} and $\const{map}_F\;\seq{\fst}\;z = x$.
    \end{tabular}
  \end{equation*}
  The third conjunct $\const{map}_F\;\seq{\snd}\;z' \sim \const{map}_F\;\seq{\embedding}\;y$ is derived similarly.

  From~\eqref{eq:relQ:alt:unfolded} to~\eqref{eq:relQ:unfolded}:
  Note that $\const{set}_{Q,i}\;[\const{map}_F\;\seq{\fst}\;z']_\sim =  \const{set}_{Q,i}\;[\const{map}_F\;\seq{\embedding}\;x]_\sim$ as $\const{map}_F\;\seq{\fst}\;z' \sim \const{map}_F\;\seq{\embedding}\;x$ are in the same equivalence class.
  Using Lemma~\ref{lemma:set:Q:natural} twice it follows that $\fst\langle \const{set}_{Q,i}\;[z']_\sim\rangle = \embedding\langle \const{set}_{Q,i}\;[x]_\sim\rangle$.
  A similar observation can be made about the second projection and $y$.
  Therefore, $\const{set}_{Q,i}\;[z']_\sim$ consists only of pairs of the form $(\embedding\;a, \embedding\;b)$.
  Accordingly, we use a function $g$ that maps $\embedding\;a$ to $a$ and $\circledast$ to some unspecified value.
  Then choose $z = \const{map}_F\;\seq{(g \times g)}\;z'$ as the witness for~\eqref{eq:relQ:unfolded},
  where $g\times g$ denotes the componentwise application of $g$ to pairs.
  We have
  \begin{equation*}
    \begin{tabular}{@{}l@{}l@{\qquad}l@{}}
      $\const{set}_{Q,i}\;[z]_\sim$ & ${} = (g \times g)\langle \const{set}_{Q,i}\;[z']_\sim\rangle$ & by Lemma~\ref{lemma:set:Q:natural} \\[0.7\jot]
      & ${} = (\embedding \times \embedding)^{-1}\langle\const{set}_{Q,i}\;[z']_\sim\rangle$ & by the above observations \\[0.7\jot]
      & ${} \subseteq (\embedding \times \embedding)^{-1}\langle\const{set}_{F,i}\;z'\rangle$ & by Theorem~\ref{thm:set:characterization} \\[0.7\jot]
      & ${} \subseteq (\embedding \times \embedding)^{-1}\langle\const{rel}_{\unit+{}}\;R_i\rangle$ & by assumption and monotonicity of preimage \\[0.7\jot]
      & ${} = R_i$. &
    \end{tabular}
  \end{equation*}
  Moreover, $\const{map}_F\;\seq{\fst}\;z = \const{map}_F\;\seq{(g \circ \fst)}\;z' \sim \const{map}_F\;\seq{(g \circ \embedding)}\;x = x$
  by \theorem{map\_comp},~\eqref{eq:map:respect} applied to $\const{map}_F\;\seq{\fst}\;z' \sim \const{map}_F\;\seq{\embedding}\;x$, and \theorem{map\_id};
  analogously for $\const{map}_F\;\seq{\snd}\;z \sim y$.
\end{proof}

\begin{exa}[Example~\ref{ex:Inl:Inl:cont3} continued]\label{ex:Inl:Inl:cont4}
  For arbitrary $x$, $y$, and $R$, we have
  \begin{equation*}
    \begin{array}{@{}r@{}c@{}l@{}}
      && ([\const{Inl}\;x]_{\sim_P}, [\const{Inl}\;y]_{\sim_P}) \in \const{rel}_Q\;R \\
      &{} \Iff {}& (\const{map}_{F_P}\;\embedding\;(\const{Inl}\;x), \const{map}_{F_P}\;\embedding\;(\const{Inl}\;y)) \in (\optparen{\sim_P} \relcomp \const{rel}_{F_P}\;(\const{rel}_{\unit+{}}\;R) \relcomp \optparen{\sim_P}) \\
      &{} \Iff {}& (\exists x'\;y'.\;\const{Inl}\;(\embedding\;x) \sim_P x' \land (x',y') \in \const{rel}_{F_P}\;(\const{rel}_{\unit+{}}\;R) \land y' \sim_P \const{Inl}\;(\embedding\;y)) \\
      &{} \Iff {}& (\exists x''\;y''.\;(\const{Inl}\;x'',\const{Inl}\;y'') \in \const{rel}_{F_P}\;(\const{rel}_{\unit+{}}\;R)) \\
      &{} \Iff {}& (\exists x''\;y''.\;(x'',y'') \in \const{rel}_{\unit+{}}\;R),
    \end{array}
  \end{equation*}
  which is always true since $(\circledast,\circledast) \in \const{rel}_{\unit+{}}\;R$.
  On the other hand,
  \begin{equation*}
    \begin{array}{@{}r@{}c@{}l@{}}
      && ([\const{Inr}\;x]_{\sim_P}, [\const{Inr}\;y]_{\sim_P}) \in \const{rel}_Q\;R \\
      &{} \Iff {}& (\const{map}_{F_P}\;\embedding\;(\const{Inr}\;x), \const{map}_{F_P}\;\embedding\;(\const{Inr}\;y)) \in (\optparen{\sim_P} \relcomp \const{rel}_{F_P}\;(\const{rel}_{\unit+{}}\;R) \relcomp \optparen{\sim_P}) \\
      &{} \Iff {}& (\exists x'\;y'.\;\const{Inr}\;(\embedding\;x) \sim_P x' \land (x',y') \in \const{rel}_{F_P}\;(\const{rel}_{\unit+{}}\;R) \land y' \sim_P \const{Inr}\;(\embedding\;y)) \\
      &{} \Iff {}& (\const{Inr}\;(\embedding\;x),\const{Inr}\;(\embedding\;y)) \in \const{rel}_{F_P}\;(\const{rel}_{\unit+{}}\;R) \\
      &{} \Iff {}& (\embedding\;x,\embedding\;y) \in \const{rel}_{\unit+{}}\;R \\
      &{} \Iff {}& (x,y) \in R.
    \end{array}
  \end{equation*}
  It is easy to see that $\const{rel}_Q\;R$ does not relate $[\const{Inl}\;x]_{\sim_P}$ and $[\const{Inr}\;y]_{\sim_P}$ or vice versa.
  Therefore, $\const{rel}_Q$ behaves exactly like the relator of $\unit + \alpha$, as expected.
  \exampleend
\end{exa}

Moreover, the following condition on $\optparen{\sim}$ characterizes when $\const{rel}_Q$ satisfies \theorem{rel\_comp}.
Again, the non-emptiness assumptions for $R_i \relcomp S_i$ come from $\const{rel}_{\unit+{}}$ extending any relation $R$ with the pair $(\circledast, \circledast)$.
\begin{equation}
  \label{eq:weak:pullback}
  (\forall i.\; R_i \relcomp S_i \neq \{\})
  \implies
  \const{rel}_F\;\seq{R} \relcomp \optparen{\sim} \relcomp \const{rel}_F\;\seq{S} \subseteq
  \optparen{\sim} \relcomp \const{rel}_F\;\seq{(R \relcomp S)} \relcomp \optparen{\sim}
\end{equation}

\looseness=-1
It turns out that this condition implies the respectfulness of the mapper~\eqref{eq:map:respect}.
Intuitively, the relator is a generalization of the mapper.
Furthermore, it is well known that subdistributivity implies preimage preservation~\cite{GummSchroeder2005AU}.
Since our conditions on $\optparen{\sim}$ characterize these preservation properties, it is no surprise that the latter implication carries over.

\begin{lem}
  For an equivalence relation $\sim$,
  condition~\eqref{eq:weak:pullback} implies respectfulness~\eqref{eq:map:respect} and preimage preservation~\eqref{eq:preimage:preservation:alt}.
\end{lem}

%\begin{rep}
\begin{proof}
  To show~\eqref{eq:map:respect}, fix $x$ and $y$ such that $x \sim y$.
  Choose the relations $R_i = \{(f_i\;a, a) \mid \const{True}\}$ and $S_i = \{(a, f_i\;a) \mid \const{True}\}$.
  Then $R_i \relcomp S_i \neq \{\}$ because types in HOL are non-empty and so is $f_i$'s image.
  We have $(\const{map}_F\;\seq{f}\;x, x) \in \const{rel}_F\;\seq{R}$ and $(y, \const{map}_F\;\seq{f}\;y) \in \const{rel}_F\;\seq{S}$ by well-known BNF properties.
  Therefore, $(\const{map}_F\;\seq{f}\;x, \const{map}_F\;\seq{f}\;y) \in (\const{rel}_F\;\seq{R} \relcomp \optparen{\sim} \relcomp \const{rel}_F\;\seq{S})$.
  Using~\eqref{eq:weak:pullback}, there exist $z$ and $z'$ such that $\const{map}_F\;\seq{f}\;x \sim z$, $(z,z') \in \const{rel}_F\;\seq{(R \relcomp S)}$, and $z' \sim \const{map}_F\;\seq{f}\;y$.
  Note that $R_i \relcomp S_i$ is equality restricted to $f_i$'s image.
  This implies $z = z'$, again by the BNF properties.
  Thus $\const{map}_F\;\seq{f}\;x \sim \const{map}_F\;\seq{f}\;y$.

  It remains to prove~\eqref{eq:preimage:preservation:alt}.
  Let $x$ and $y$ be such that $\const{map}_F\;\seq{f}\;x \sim y$ and $y \in \internalize{F}\;\seq{A}$.
  Choose $R_i = A_i \times A_i$ and $S_i = \{ (f_i\; a, a) \mid \const{True} \}$.
  Then $R_i \relcomp S_i \neq \{\}$ as $f_i^{-1} \langle A_i \rangle \neq \{\}$ by assumption in~\eqref{eq:preimage:preservation:alt}.
  Moreover, $(y, y) \in \const{rel}_F\;\seq{R}$ as $\const{map}_F\;\seq{(\lambda a.\; (a, a))}\;y \in \internalize{F}\;\seq{R}$.
  Further, $(\const{map}_F\;\seq{f}\;x, x) \in \const{rel}_F\;\seq{S}$ as $\const{map}_F\;\seq{(\lambda a.\; (f\;a, a))}\;x \in \internalize{F}\;\seq{S}$.
  Therefore, $(y, x) \in \optparen{\sim} \relcomp \const{rel}_F\;\seq{(R \relcomp S)} \relcomp \optparen{\sim}$ by~\eqref{eq:weak:pullback}.
  So there are $u$ and $v$ such that $y \sim u$, $(u, v) \in \const{rel}_F\;\seq{(R \relcomp S)}$, and $v \sim x$.
  By the BNF properties, there is a $w$ such that $(u, w) \in \const{rel}_F\;\seq{R}$ and $(w, v) \in \const{rel}_F\;\seq{S}$.
  Then, $w \in \internalize{F}\;\seq{A}$ and $v = \const{map}_F\;\seq{f}\;w$.
  So $v \in (\const{map}_F\;\seq{f})^{-1}\langle \internalize{F}\;\seq{A} \rangle$ and $x \sim v$.
\end{proof}
%\end{rep}

In summary, we obtain the following main preservation theorem:
\begin{thm}%
  \label{thm:bnf:quotient}
  The quotient $\seq{\alpha}\;Q = \seq{\alpha}\;F / \quotparen{\sim}$ inherits the structure from the BNF $\seq{\alpha}\;F$ with the mapper $\const{map}_Q\;\seq{f}\;[x]_\sim = [\const{map}_F\;\seq{f}\;x]_\sim$
 if $\optparen{\sim}$ satisfies the conditions~\eqref{eq:equiv:sim},~\eqref{eq:wide:intersections}, and~\eqref{eq:weak:pullback}.
 The setters and relator are given by Theorems~\ref{thm:set:characterization} and~\ref{thm:rel:characterization}, respectively.
\end{thm}

\begin{exa}%
    \label{ex:tllist}
    A terminated coinductive list $(\alpha,\beta)\;\tc{tllist}$
    is either a finite list of $\alpha$ values terminated by a single $\beta$ value, or an infinite list of $\alpha$ values.
    This type can be seen as a quotient of pairs $\alpha\;\tc{llist} \times \beta$,
    where the first component stores the possibly infinite list given by a codatatype $\tc{llist}$ and the second component stores the terminator.
    The equivalence relation identifies all pairs with the same infinite list in the first component, effectively removing the terminator from infinite lists.\footnote{%
      \looseness=-1
      Clearly, $\tc{tllist}$ could be defined directly as a codatatype.
      When Isabelle had no codatatype command, one of the authors formalized $\tc{tllist}$ via this quotient~\cite[version for Isabelle2013]{Coinductive2013}.
    }
    Let $(\var{xs},b) \sim_{\tc{tllist}} (\var{ys},c)$ iff $\var{xs} = \var{ys}$ and, if $\var{xs}$ is finite, $b = c$.
    Like $\optparen{\sim_P}$ from Example~\ref{ex:Inl:Inl}, $\optparen{\sim_{\tc{tllist}}}$ does not satisfy the naive condition~\eqref{eq:set:respect:naive}.
    \begin{quotex}
        $\keyw{codatatype }\alpha\;\tc{llist} = \const{LNil} \mid \const{LCons}\;\alpha\;(\alpha\;\tc{llist})$
        \\
        $\keyw{quotient\_type }(\alpha,\beta)\;\tc{tllist} = (\alpha\;\tc{llist} \times \beta) / \quotparen{\sim_{\tc{tllist}}}$
    \end{quotex}

    Our goal is the construction of (co)datatypes with recursion through quotients such as $(\alpha,\beta)\;\tc{tllist}$.
    As a realistic example, consider an inductive model of a finite interactive system that produces a possibly unbounded sequence of outputs $\tc{out}$ for every input $\tc{in}$:
    \begin{quotex}
        $\keyw{datatype }\tc{system} = \const{Step}\;(\tc{in} \fun (\tc{out}, \tc{system})\;\tc{tllist})$
    \end{quotex}
    This datatype definition is only possible if $\tc{tllist}$ is a BNF in $\beta$.
    Previously, this had to be shown by manually defining the mapper and setters and proving the BNF properties.
    Theorem~\ref{thm:bnf:quotient} identifies the conditions under which $\tc{tllist}$ inherits the BNF structure of its underlying type, and it allows us to automate these definitions and proofs.
    For $\tc{tlllist}$, the conditions can be discharged easily using automatic proof methods and a simple lemma about $\tc{llist}$'s relator that states that related lists are either both finite or infinite.
    \exampleend

%     First, we define coinductive (possibly infinite) lists as a codatatype $\tc{llist}$.
%     Like all (co)datatypes, $\tc{llist}$ is a BNF.

%     We then attempt to define $(\alpha,\beta)\;\tc{tllist}$ as the product
%     However, two infinite $\tc{tllist}$s should not be distinguishable by their $\beta$ values alone, as they do not have a terminator.
%     Therefore, we use a quotient that identifies all pairs with the same infinite list in the first component, effectively discarding the second component.
%     Let $(\var{xs},b) \sim_{\tc{tllist}} (\var{ys},c)$ iff $\var{xs} = \var{ys}$ and, if $\var{xs}$ is finite, $b = c$.
%     Note that $\optparen{\sim_{\tc{tllist}}}$, like $\optparen{\sim_P}$ from Example~\ref{ex:Inl:Inl}, does not satisfy the naive condition~\eqref{eq:set:respect:naive}.
%     \begin{quotex}
%         $\keyw{quotient\_type }(\alpha,\beta)\;\tc{tllist} = (\alpha\;\tc{llist} \times \beta) / \quotparen{\sim_{\tc{tllist}}}$
%     \end{quotex}
%     %\pagebreak[2]
\end{exa}

\subsection{Subdistributivity via confluent relations}%
\label{section:confluence}

Among the BNF properties, subdistributivity (\theorem{rel\_comp}) is typically the hardest to show.
For example, distinct lists (type $\alpha\;\tc{dlist}$) have been shown to be a BNF\@.
The manual proof requires 126 lines.
Of these, the subdistributivity proof takes about 100 lines.
Yet, with the theory developed so far, essentially the same argument is needed for the subdistributivity condition~\eqref{eq:weak:pullback}.
We now present a sufficient criterion for subdistributivity that simplifies such proofs.
For $\tc{dlist}$, this shortens the subdistributivity proof to 58~lines.
With our \keyw{lift\_bnf} command (Section~\ref{section:implementation}), the whole proof is now 64~lines, half of the manual proof.
% AL: I haven't counted in those lines the proof of the quotient theorem.

Equivalence relations are often (or can be) expressed as the equivalence closure of a rewrite relation $\optparen{\rightsquigarrow}$.
For example, the subdistributivity proof for distinct lists views $\alpha\;\tc{dlist}$ as the quotient $\alpha\;\tc{list} / \quotparen{\sim_{\tc{dlist}}}$ with $\mathit{xs} \sim_{\tc{dlist}} \mathit{ys}$ iff $\const{remdups}\;\mathit{xs} = \const{remdups}\;\mathit{ys}$, where $\const{remdups}\;\mathit{xs}$ keeps only the last occurrence of every element in $\mathit{xs}$.
So, $\optparen{\sim_{\tc{dlist}}}$ is the equivalence closure of the following relation $\optparen{\rightsquigarrow_{\tc{dlist}}}$, where $\optparen{\append}$ concatenates two lists:
\begin{equation*}
  \label{eq:rewrite:dlist}
  \mathit{xs} \append [x] \append \mathit{ys} \rightsquigarrow_{\tc{dlist}} \mathit{xs} \append \mathit{ys}
  \text{ if }
  x \in \const{set}\;\mathit{ys}
\end{equation*}
We use the following notation:
$\optparen{\leftsquigarrow}$ denotes the reverse relation, i.e., $x \leftsquigarrow y$ iff $y \rightsquigarrow x$.
Further, % $(\leftrightsquigarrow)$ denotes the symmetric closure,
%$(\tranclp{\rightsquigarrow})$ the transitive closure,
$\optparen{\rtranclp{\rightsquigarrow}}$ denotes the reflexive and transitive closure, and $\optparen{\rtranclp{\leftrightsquigarrow}}$ the equivalence closure.
A relation $\optparen{\rightsquigarrow}$ is confluent iff whenever $x \rtranclp{\rightsquigarrow} y$ and $x \rtranclp{\rightsquigarrow} z$, then there exists a $u$ such that $y \rtranclp{\rightsquigarrow} u$ and $z \rtranclp{\rightsquigarrow} u$---or, equivalently in pointfree style, if $(\rtranclp{\leftsquigarrow} \relcomp \rtranclp{\rightsquigarrow}) \subseteq (\rtranclp{\rightsquigarrow} \relcomp \rtranclp{\leftsquigarrow})$.

\begin{thm}[Subdistributivity via confluent relations]%
  \label{thm:confluent:quotient}
  Let an equivalence relation $\optparen{\sim}$ satisfy~\eqref{eq:map:respect}.
  Then, it also satisfies~\eqref{eq:weak:pullback} if there is a confluent relation $\optparen{\rightsquigarrow}$ with the following properties:
  \begin{enumerate}[label=(\roman*)]
  \item\label{item:eq:in:rtcs}
    The equivalence relation is contained in $\optparen{\rightsquigarrow}$'s equivalence closure:
    $(\sim) \subseteq (\rtranclp{\leftrightsquigarrow})$.\footnote{%
      The other inclusion $(\rtranclp{\leftrightsquigarrow}) \subseteq (\sim)$ follows from the second condition with~\eqref{eq:map:respect}:
      For $x \rightsquigarrow y$, let $x' = \const{map}_F\;\seq{(\lambda a.\;(a, a))}\;x$.
      As $x = \const{map}_F\;\seq{\const{fst}}\;x'$, there exists a $y'$ with $x' \sim y'$ and $y = \const{map}_F\;\seq{\const{fst}}\;y'$.
      Applying~\eqref{eq:map:respect} to $x' \sim y'$, we have $x = \const{map}_F\;\seq{\const{fst}}\;x' \sim \const{map}_F\;\seq{\const{fst}}\;y' = y$.
      So $(\rightsquigarrow) \subseteq (\sim)$ and therefore $(\rtranclp{\leftrightsquigarrow}) \subseteq (\sim)$.
    }
  \item\label{item:projection}
    The relation factors through projections:
    if $\const{map}_F\;\seq{\fst}\;x \rightsquigarrow y$
    then there exists a $y'$ such that $y = \const{map}_F\;\seq{\fst}\;y'$ and $x \sim y'$ and $\const{set}_{F,i}\;y' \subseteq \const{set}_{F,i}\;x$ for all $i$;
    and similarly for $\snd$.
  \end{enumerate}
\end{thm}

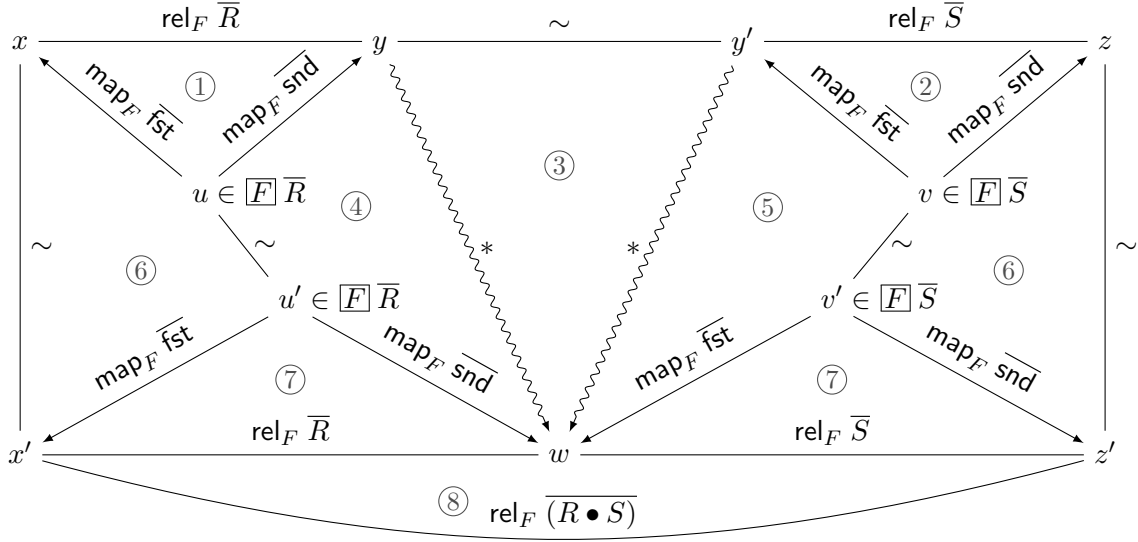
\begin{figure}[t]
    \centering
    \begin{tikzpicture}[node distance=4.3cm, >=latex]
    \begin{scope}[font=\vphantom{Ag}]
    \node (x) { $x$ };
    \node (y) [right=of x] { $y$ };
    \node (y') [right=of y] { $y'$ };
    \node (z) [right=of y'] { $z$ };
    \end{scope}

    \begin{scope}[anchor=south]
    \draw (x) -- (y)
    node (relxy) [midway] { $\const{rel}_F\;\seq{R}$ };
    \draw (y) -- (y')
    node (simyy') [midway] { $\sim$ };
    \draw (y') -- (z)
    node (rely'z) [midway] { $\const{rel}_F\;\seq{S}$ };
    \end{scope}

    \begin{scope}[font=\vphantom{Ag}]
    \node (u) at ([yshift=-2cm] relxy.south) { $u$\rlap{${} \in \internalize{F}\;\seq{R}$} };
    \node (v) at (u -| rely'z) { $v$\rlap{${} \in \internalize{F}\;\seq{S}$} };
    \end{scope}

    \begin{scope}[sloped, anchor=south]
    \draw[->] (u) -- (x) node (mapux) [pos=0.47] { $\const{map}_F\;\seq{\fst}$ };
    \draw[->] (u) -- (y) node (mapuy) [pos=0.47] { $\const{map}_F\;\seq{\snd}$ };
    \draw[->] (v) -- (y') node (mapvy') [pos=0.47] { $\const{map}_F\;\seq{\fst}$ };
    \draw[->] (v) -- (z) node (mapvz) [pos=0.47] { $\const{map}_F\;\seq{\snd}$ };
    \end{scope}

    \begin{scope}[font=\vphantom{Ag}]
    \node (w) at ([yshift=-5.5cm] simyy'.south) { $w$ };
    \node (x') at (w -| x) { $x'$ };
    \node (z') at (w -| z) { $z'$ };
    \end{scope}

    \begin{scope}[anchor=south]
    \draw (x') -- (w)
    node (relx'w) [midway] { $\const{rel}_F\;\seq{R}$ };
    \draw (w) -- (z')
    node (relwz') [midway] { $\const{rel}_F\;\seq{S}$ };
    \end{scope}

    \begin{scope}[font=\vphantom{Ag},text depth=1.5ex]
    \node (u') at ([yshift=2cm] relx'w.south) { $u'$\rlap{${} \in \internalize{F}\;\seq{R}$}};
    \node (v') at (u' -| relwz') { $v'$\rlap{${} \in \internalize{F}\;\seq{S}$}};
    \end{scope}

    \begin{scope}[decoration={snake, amplitude=1.25pt, segment length=5pt, post length=5pt}]
    \draw[->, decorate] (y) -- (w) node [midway, anchor=west] { $\ast$ };
    \draw[->, decorate] (y') -- (w) node [midway, anchor=east] { $\ast$ };
    \end{scope}

    \begin{scope}[sloped, anchor=south]
    \draw[->] (u') -- (x') node (mapu'x) [midway] { $\const{map}_F\;\seq{\fst}$ };
    \draw[->] (u') -- (w) node (mapu'y) [midway] { $\const{map}_F\;\seq{\snd}$ };
    \draw[->] (v') -- (w) node (mapv'y') [midway] { $\const{map}_F\;\seq{\fst}$ };
    \draw[->] (v') -- (z') node (map'vz) [midway] { $\const{map}_F\;\seq{\snd}$ };
    \end{scope}

    \begin{scope}[anchor=west]
    \draw (u') -- (u) node (simuu') [midway] { $\sim$ };
    \draw (v') -- (v) node (simvv') [midway] { $\sim$ };
    \draw (x) -- (x') node (simxx') [midway] { $\sim$ };
    \draw (z) -- (z') node (simzz') [midway, overlay] { $\sim$ };
    \end{scope}

    \draw (x') to[out=-15, in=-165]
    node (relx'z') [midway, anchor=south] { $\const{rel}_F\;\seq{(R \relcomp S)}$ }
    (z');

    \begin{scope}[circlednode]
    \node at ($ (relxy.south) ! 0.3 ! (u) $) {1};
    \node at ($ (rely'z.south) ! 0.3 ! (v) $) {2};
    \node at ($ (simyy'.south) ! 0.3 ! (w) $) {3};
    \node at ($ (mapuy) ! 0.3 ! (w) $) {4};
    \node at ($ (mapvy') ! 0.3 ! (w) $) {5};
    \node at ($ (u') ! 0.6 ! (simxx') $) {6};
    \node at ($ (v') ! 0.6 ! (simzz') $) {6};
    \node at ($ (u') ! 0.6 ! (relx'w) $) {7};
    \node at ($ (v') ! 0.6 ! (relwz') $) {7};
    \node at (relx'z'.west) [xshift=-0.2cm, anchor=south east] {8};
    \end{scope}
    \end{tikzpicture}
    \vspace*{-3ex}
    \caption{Proof diagram for Theorem~\ref{thm:confluent:quotient}}%
    \label{fig:confluent:quotient:proof}
\end{figure}

\begin{proof}
  The proof is illustrated in Fig.~\ref{fig:confluent:quotient:proof}.
  The proof starts at the top with $(x, z) \in (\const{rel}_F\;\seq{R} \relcomp \optparen{\sim} \relcomp \const{rel}_F\;\seq{S})$, i.e., there are $y$ and $y'$ such that $(x, y) \in \const{rel}_F\;\seq{R}$ and $y \sim y'$ and $(y', z) \in \const{rel}_F\;\seq{S}$.
  We show $(x, z) \in (\optparen{\sim} \relcomp \const{rel}_F\;\seq{(R \relcomp S)} \relcomp \optparen{\sim})$ by establishing the path from $x$ to $z$ via $x'$ and $z'$ along the three other borders of the diagram.

  First \circled{1}, by \theorem{in\_rel}, there is a $u \in \internalize{F}\;\seq{R}$ such that $x = \const{map}_F\;\seq{\fst}\;u$ and $y = \const{map}_F\;\seq{\snd}\;u$.
  Similarly, $\const{rel}_F\;\seq{S}\;y'\;z$ yields a $v$ with the corresponding properties \circled{2}.

  Second, by~\ref{item:eq:in:rtcs}, $y \sim y'$ implies $y \rtranclp{\leftrightsquigarrow} y'$.
  Since $\optparen{\rightsquigarrow}$ is confluent, there exists a $w$ such that $y \rtranclp{\rightsquigarrow} w$ and $y' \rtranclp{\rightsquigarrow} w$ \circled{3}.
  By induction on $\optparen{\rtranclp{\rightsquigarrow}}$ using~\ref{item:projection},
  $y \rtranclp{\rightsquigarrow} w$ factors through the projection $y = \const{map}_F\;\seq{\snd}\;u$ and we obtain a $u'$ such that $u \sim u'$ and $w = \const{map}_F\;\seq{\snd}\;u'$ and $\const{set}_{F,i}\;u' \subseteq \const{set}_{F,i}\;u$ for all $i$~\circled{4}.
  Analogously, we obtain $v'$ corresponding to $y'$ and $v$ \circled{5}.
  Set $x' = \const{map}_F\;\seq{\fst}\;u'$ and $z' = \const{map}_F\;\seq{\snd}\;v'$.
  As $\const{map}_F$ preserves $\optparen{\sim}$ by~\eqref{eq:map:respect}, we have $x \sim x'$ and $z \sim z'$ \circled{6}.

  Next, we focus on the two triangles at the bottom \circled{7}.
  By defininition of $\internalize{F}$, $\const{set}_{F,i}\;u' \subseteq \const{set}_{F,i}\;u$ for all $i$ and $u \in \internalize{F}\;\seq{R}$ imply $u' \in \internalize{F}\;\seq{R}$; similarly $v' \in \internalize{F}\;\seq{S}$.
  Now, $u'$ and $v'$ are the witnesses to the existential in
  \theorem{in\_rel} for $x'$ and $w$, and $w$ and $z'$, respectively.
  So $(x', w) \in \const{rel}_F\;\seq{R}$ and $(w, z') \in \const{rel}_F\;\seq{S}$, i.e., $(x', z') \in (\const{rel}_F\;\seq{R} \relcomp \const{rel}_F\;\seq{S})$.
  Finally, as $F$ is a BNF, $(x', z') \in \const{rel}_F\;\seq{(R \relcomp S)}$ follows with subdistributivity \theorem{rel\_comp} \circled{8}.
\end{proof}

\begin{exa}\label{ex:dlist}
  \looseness=-1
  For distinct lists, we have $(\sim_{\tc{dlist}}) = (\rtranclp{\leftrightsquigarrow}_{\tc{dlist}})$
  and $\optparen{\rightsquigarrow_{\tc{dlist}}}$ is confluent.
  However, condition~\ref{item:projection} of Theorem~\ref{thm:confluent:quotient} does not hold.
  For example, for $x = [(1, a), (1, b)]$, we have $\const{map}_{\tc{list}}\;\fst\;x = [1, 1] \rightsquigarrow_{\tc{dlist}} [1]$,
  and yet there is no $y$ such that $x \sim_{\tc{dlist}} y$ and $\const{map}_{\tc{list}}\;\fst\;y = [1]$.
  The problem is that the projection $\const{map}_{\tc{list}}\;\fst$ makes different atoms of $x$ equal and $\optparen{\rightsquigarrow_{\tc{dlist}}}$ \emph{removes} equal atoms, but the removal cannot be mimicked on $x$ itself.
  Fortunately, we can also \emph{add} equal atoms instead of removing them.
  Define $\optparen{\rightsquigarrow'_{\tc{dlist}}}$ by
  \begin{equation*}
     \mathit{xs} \append \mathit{ys} \rightsquigarrow'_{\tc{dlist}} \mathit{xs} \append [x] \append \mathit{ys}
     \text{ if }
     x \in \const{set}\;\mathit{ys}
  \end{equation*}
  Then, $\optparen{\rightsquigarrow'_{\tc{dlist}}}$ is confluent and factors through projections.
  So distinct lists inherit the BNF structure from lists by Theorem~\ref{thm:confluent:quotient} and either Lemma~\ref{lem:wide:intersections} or Lemma~\ref{lemma:wide:intersection:finite}.
  \exampleend
\end{exa}

\begin{exa}
  The free monoid over atoms $\alpha$ consists of all finite lists $\alpha\;\tc{list}$.
  The free idempotent monoid $\alpha\;\tc{fim}$ is then the quotient $\alpha\;\tc{list} / \quotparen{\sim_{\tc{fim}}}$ where $\optparen{\sim_{\tc{fim}}}$ is the equivalence closure of the idempotence law for list concatenation
  \begin{equation*}
    \mathit{xs} \append \mathit{ys} \append \mathit{zs}
    \rightsquigarrow_{\tc{fim}}
    \mathit{xs} \append \mathit{ys} \append \mathit{ys} \append \mathit{zs}
  \end{equation*}
  \looseness=-1
  We have oriented the rule such that it introduces rather than removes the duplication.
  In term rewriting, the rule is typically oriented in the other direction~\cite{Hullot1980SRI}
  such that the resulting rewriting system terminates;
  however, this classical relation $\optparen{\leftsquigarrow_{\tc{fim}}}$ is not confluent:
  \emph{ababcbabc} has two normal forms $\emph{a\underline{babcbabc}} \leftsquigarrow_{\tc{fim}} \emph{\underline{abab}c} \leftsquigarrow_{\tc{fim}} \emph{abc}$ and $\emph{\underline{abab}cbabc} \leftsquigarrow_{\tc{fim}} \emph{abcbabc}$ (redexes are underlined).
%   \footnote{
%     The following shows the reductions, where the eliminated repetition is underlined.
%     \emph{\underline{abcacbabcabcacbabc}babacbcacabc}
%     $\leftsquigarrow_{\tc{fim}}$ \emph{abca\underline{cbabcbab}acbcacabc}
%     $\leftsquigarrow_{\tc{fim}}$ \emph{abcac\underline{baba}cbcacabc}
%     $\leftsquigarrow_{\tc{fim}}$ \emph{abc\underline{acbacb}cacabc}
%     $\leftsquigarrow_{\tc{fim}}$ \emph{a\underline{bcacbcac}abc}
%     $\leftsquigarrow_{\tc{fim}}$ \emph{ab\underline{caca}bc}
%     $\leftsquigarrow_{\tc{fim}}$ \emph{\underline{abcabc}}
%     $\leftsquigarrow_{\tc{fim}}$ \emph{abc}
%     and
%     \emph{abcacbabcabca\underline{cbabcbab}acbcacabc}
%     $\leftsquigarrow_{\tc{fim}}$ \emph{abcacbabcabcac\underline{baba}cbcacabc}
%     $\leftsquigarrow_{\tc{fim}}$ \emph{abcacbabcabc\underline{acbacb}cacabc}
%     $\leftsquigarrow_{\tc{fim}}$ \emph{abcacbabca\underline{bcacbcac}abc}
%     $\leftsquigarrow_{\tc{fim}}$ \emph{abcacbabcab\underline{caca}bc}
%     $\leftsquigarrow_{\tc{fim}}$ \emph{abcacbabc\underline{abcabc}}
%     $\leftsquigarrow_{\tc{fim}}$ \emph{abcacb\underline{abcabc}}
%     $\leftsquigarrow_{\tc{fim}}$ \emph{abcacbabc}.
%   }
  In contrast, our orientation yields a confluent relation $\optparen{\rightsquigarrow_{\tc{fim}}}$,
  although the formal proof requires some effort.
  The relation also factors through projections.
  So by Theorem~\ref{thm:confluent:quotient} and either Lemma~\ref{lem:wide:intersections} or Lemma~\ref{lemma:wide:intersection:finite}, the free idempotent monoid $\alpha\;\tc{fim}$ is also a BNF\@.
  \exampleend
\end{exa}

\begin{exa}
  A cyclic list is a finite list where the two ends are glued together.
  Abbot et al.~\cite{AbbottAltenkirchGhaniMcBride2004MPC} define the type of cyclic lists as the quotient that identifies lists whose elements have been shifted.
  Let $\optparen{\rightsquigarrow_{\tc{rotate}}}$ denote the one-step rotation of a list, i.e.,
  \begin{equation*}
    [] \rightsquigarrow_{\tc{rotate}} []
    \qquad
    \qquad
    [x] \append \mathit{xs} \rightsquigarrow_{\tc{rotate}} \mathit{xs} \append [x]
  \end{equation*}
  The quotient $\alpha\;\tc{cyclist} = \alpha\;\tc{list} / \quotparen{\rtranclp{\leftrightsquigarrow}_{\tc{rotate}}}$ is a BNF as $\optparen{\rightsquigarrow_{\tc{rotate}}}$ satisfies the conditions of Theorem~\ref{thm:confluent:quotient} and Lemmas~\ref{lem:wide:intersections} and~\ref{lemma:wide:intersection:finite}.
  \exampleend
\end{exa}

\begin{exa}[Example~\ref{ex:regex} continued]
We prove the fact that $\alpha\;\tc{re}_\tc{aci}$ is a BNF using Theorem~\ref{thm:confluent:quotient}. The
confluent rewrite relation $\optparen{\rightsquigarrow_\tc{aci}}$ that satisfies the conditions of Theorem~\ref{thm:confluent:quotient} and whose equivalence closure is $\optparen{\sim_{\tc{aci}}}$ is defined inductively as follows.

\begin{quotex}
\begin{tabular}[b]{@{}l@{\qquad}l@{}}
$\const{Alt}\;(\const{Alt}\;r\;s)\;t \rightsquigarrow_\tc{aci} \const{Alt}\;r\;(\const{Alt}\;s\;t)$&
$\const{Alt}\;r\;(\const{Alt}\;s\;t) \rightsquigarrow_\tc{aci} \const{Alt}\;(\const{Alt}\;r\;s)\;t$\\[0.7\jot]
$\const{Alt}\;r\;s \rightsquigarrow_\tc{aci} \const{Alt}\;s\;r$&
$r \rightsquigarrow_\tc{aci} \const{Alt}\;r\;r$\\[0.7\jot]
$r \rightsquigarrow_\tc{aci} r$\\[0.7\jot]
%$r\rightsquigarrow_\tc{aci}r'\implies \const{Alt}\;r\;s \rightsquigarrow_\tc{aci} \const{Alt}\;r'\;s$&
%$s\rightsquigarrow_\tc{aci}s'\implies \const{Alt}\;r\;s \rightsquigarrow_\tc{aci} \const{Alt}\;r\;s'$\\
\multicolumn{2}{@{}l@{}}{
$r\rightsquigarrow_\tc{aci}r'\implies s\rightsquigarrow_\tc{aci}s'\implies \const{Alt}\;r\;s \rightsquigarrow_\tc{aci} \const{Alt}\;r'\;s'$}\\[0.7\jot]
%$r\rightsquigarrow_\tc{aci}r'\implies \const{Conc}\;r\;s \rightsquigarrow_\tc{aci} \const{Conc}\;r'\;s$&
%$s\rightsquigarrow_\tc{aci}s'\implies \const{Conc}\;r\;s \rightsquigarrow_\tc{aci} \const{Conc}\;r\;s'$\\
\multicolumn{2}{@{}l@{}}{
$r\rightsquigarrow_\tc{aci}r'\implies s\rightsquigarrow_\tc{aci}s'\implies \const{Conc}\;r\;s \rightsquigarrow_\tc{aci} \const{Conc}\;r'\;s'$}\\[0.7\jot]
$r\rightsquigarrow_\tc{aci}r'\implies \const{Star}\;r \rightsquigarrow_\tc{aci} \const{Star}\;r'$
\end{tabular}%
\exampleend
\end{quotex}
\end{exa}

\begin{exa}\label{ex:regex2}
We now consider a variant of Example~\ref{ex:regex}, where we quotient the regular expressions by $\sim_\tc{acidz}$, the least equivalence relation satisfying
\begin{quotex}
\begin{tabular}{@{}l@{\qquad}l@{\qquad}l@{}}
$\const{Alt}\;(\const{Alt}\;r\;s)\;t \sim_\tc{acidz} \const{Alt}\;r\;(\const{Alt}\;s\;t)$
&
$\const{Alt}\;r\;s \sim_\tc{acidz} \const{Alt}\;s\;r$
&
$\const{Alt}\;r\;r \sim_\tc{acidz} r$\\[0.7\jot]
$\const{Conc}\;\const{Zero}\;r \sim_\tc{acidz} \const{Zero}$&
$\const{Alt}\;\const{Zero}\;r \sim_\tc{acidz} r$\\[0.7\jot]
\multicolumn{2}{@{}l@{}}{$\const{Conc}\;(\const{Alt}\;r\;s)\;t \sim_\tc{acidz} \const{Alt}\;(\const{Conc}\;r\;t)\;(\const{Conc}\;s\;t)$}\\[0.7\jot]
\multicolumn{2}{@{}l@{}}{
    $r\sim_\tc{acidz}r'\implies s\sim_\tc{acidz}s'\implies \const{Alt}\;r\;s \sim_\tc{acidz} \const{Alt}\;r'\;s'$}\\[0.7\jot]
\multicolumn{2}{@{}l@{}}{
    $r\sim_\tc{acidz}r'\implies \const{Conc}\;r\;s \sim_\tc{acidz} \const{Conc}\;r'\;s$}
\end{tabular}
\end{quotex}

\noindent
This relation is of special interest because, like $\optparen{\sim_\tc{aci}}$, it gives rise to a finite automaton construction when computing with Brzozowski derivatives of regular expressions equated by $\optparen{\sim_{\tc{acidz}}}$ and, moreover, the constructed automaton is isomorphic to the one obtained by the subset construction via Antimirov's partial derivatives~\cite[Section 4.2]{DBLP:conf/itp/NipkowT14}.

A first observation is that, unlike $\optparen{\sim_{\tc{aci}}}$, the relation $\optparen{\sim_{\tc{acidz}}}$ does not preserve the regular expression's setter: the offending rule being $\const{Conc}\;\const{Zero}\;r \sim_\tc{acidz} \const{Zero}$. Nevertheless, $\optparen{\sim_{\tc{acidz}}}$ satisfies $\eqref{eq:map:respect}$, which together with the fact that regular expressions are finite objects (and thus have finitely many atoms) allows us to use Lemma~$\ref{lemma:wide:intersection:finite}$ to obtain~\eqref{eq:wide:intersection:one}.

\looseness=-1
As before, we aim to apply Theorem~\ref{thm:confluent:quotient} to show subdistributivity. Notably, the above relation is not a congruence: one is not allowed to apply the equivalence under the $\const{Star}$ constructor and in the second argument of $\const{Conc}$. Thus, our confluent ``rewrite'' relation $\optparen{\rightsquigarrow_{\tc{acidz}}}$ also inherits these constraints. Moreover, it uses two auxiliary functions $\const{elim\_zeros}$ and $\const{distribute}$ to handle the new (compared to $\optparen{\sim_\tc{aci}}$) cases related to the $\const{Zero}$ and $\const{Conc}$ operators.
\begin{quotex}
\begin{tabular}{@{}l@{\qquad}l@{}}
\multicolumn{2}{@{}l@{}}{$\const{elim\_zeros}\;(\const{Alt}\;r\;s) = \syntax{let}\;r'=\const{elim\_zeros}\;r;\;s'=\const{elim\_zeros}\;s\;\syntax{in}$}\\[0.7\jot]
\multicolumn{2}{@{}l@{}}{$\phantom{\const{elim\_zeros}\;(\const{Alt}\;r\;s) ={}}\syntax{if}\;r'=\const{Zero}\;\syntax{then}\;s'\;\syntax{else}\;\syntax{if}\;s'=\const{Zero}\;\syntax{then}\;r'\;\syntax{else}\;\const{Alt}\;r'\;s'$}\\[0.7\jot]
\multicolumn{2}{@{}l@{}}{$\const{elim\_zeros}\;(\const{Conc}\;r\;s) = \syntax{let}\;r'=\const{elim\_zeros}\;r\;\syntax{in}\;\syntax{if}\;r'=\const{Zero}\;\syntax{then}\;\const{Zero}\;\syntax{else}\;\const{Conc}\;r'\;s$}\\[0.7\jot]
\multicolumn{2}{@{}l@{}}{$\const{elim\_zeros}\;r = r$}\\[4\jot]

\multicolumn{2}{@{}l@{}}{$\const{distribute}\;t\;(\const{Alt}\;r\;s) = \const{Alt}\;(\const{distribute}\;t\;r)\;(\const{distribute}\;t\;s)$}\\[0.7\jot]
\multicolumn{2}{@{}l@{}}{$\const{distribute}\;t\;(\const{Conc}\;r\;s) = \const{Conc}\;(\const{distribute}\;t\;r)\;s$}\\[0.7\jot]
\multicolumn{2}{@{}l@{}}{$\const{distribute}\;t\;r = \const{Conc}\;r\;t$}\\[4\jot]
$\const{Alt}\;(\const{Alt}\;r\;s)\;t \rightsquigarrow_\tc{acidz} \const{Alt}\;r\;(\const{Alt}\;s\;t)$&
$\const{Alt}\;r\;(\const{Alt}\;s\;t) \rightsquigarrow_\tc{acidz} \const{Alt}\;(\const{Alt}\;r\;s)\;t$\\[0.7\jot]
$\const{Alt}\;r\;s \rightsquigarrow_\tc{acidz} \const{Alt}\;s\;r$&
$r \rightsquigarrow_\tc{acidz} \const{Alt}\;r\;r$\\[0.7\jot]
$r \rightsquigarrow_\tc{acidz} s\implies r\rightsquigarrow_\tc{acidz}\const{elim\_zeros}\;s$& $\const{Conc}\;r\;s\rightsquigarrow_\tc{acidz}\const{distribute}\;s\;r$
\\[0.7\jot]
$r \rightsquigarrow_\tc{acidz} r$\\[0.7\jot]
%$r\rightsquigarrow_\tc{acidz}r'\implies \const{Alt}\;r\;s \rightsquigarrow_\tc{acidz} \const{Alt}\;r'\;s$&
%$s\rightsquigarrow_\tc{acidz}s'\implies \const{Alt}\;r\;s \rightsquigarrow_\tc{acidz} \const{Alt}\;r\;s'$\\
\multicolumn{2}{@{}l@{}}{
    $r\rightsquigarrow_\tc{acidz}r'\implies s\rightsquigarrow_\tc{acidz}s'\implies \const{Alt}\;r\;s \rightsquigarrow_\tc{acidz} \const{Alt}\;r'\;s'$}\\[0.7\jot]
%$r\rightsquigarrow_\tc{acidz}r'\implies \const{Conc}\;r\;s \rightsquigarrow_\tc{acidz} \const{Conc}\;r'\;s$&
%$s\rightsquigarrow_\tc{acidz}s'\implies \const{Conc}\;r\;s \rightsquigarrow_\tc{acidz} \const{Conc}\;r\;s'$\\
\multicolumn{2}{@{}l@{}}{
    $r\rightsquigarrow_\tc{acidz}r'\implies \const{Conc}\;r\;s \rightsquigarrow_\tc{acidz} \const{Conc}\;r'\;s$}
\end{tabular}
\end{quotex}

Showing confluence of $\optparen{\rightsquigarrow_{\tc{acidz}}}$ is challenging because $\optparen{\rightsquigarrow_{\tc{acidz}}}$ does not terminate due to $r \rightsquigarrow_\tc{acidz} \const{Alt}\;r\;r$.
In fact, we proceed
by showing strong confluence. The above rules are carefully designed to be restrictive in what can be rewritten in a single step (fewer restrictions would result in more cases that must be considered) and just permissive enough to be able to join the
critical pairs by performing only a single step on one side. In contrast,
establishing that $\optparen{\sim_{\tc{acidz}}}$ is the equivalence closure of
$\optparen{\rightsquigarrow_{\tc{acidz}}}$ is routine. The missing bit that
$\optparen{\rightsquigarrow_{\tc{acidz}}}$ factors through projections is also a
straightforward induction proof, after showing that the sets of atoms can only decrease along
$\optparen{\rightsquigarrow_{\tc{acidz}}}$ (which happens when $\const{elim\_zeros}$ is used).
Altogether, $\optparen{\rightsquigarrow_{\tc{acidz}}}$ satisfies the assumption of Theorem~\ref{thm:confluent:quotient} and allows us to lift the BNF structure of $\alpha\;\tc{re}$ to the quotient type:
\begin{quotex}
\keyw{quotient\_type}\;$\alpha\;\tc{re}_\tc{acidz} = \alpha\;\tc{re} / \quotparen{\sim_\tc{acidz}}$
\end{quotex}

We note that the
Theorem~\ref{thm:confluent:quotient}'s counterpart from our conference paper~\cite[Theorem
4]{DBLP:conf/cade/FurerLST20} does not apply to
$\optparen{\rightsquigarrow_{\tc{acidz}}}$ as it required the rewrite relation
to preserve the setters.\exampleend
\end{exa}

%\begin{rep}
\subsection{Non-emptiness witnesses}%
\label{section:witness}

\looseness=-1
An often neglected (also by us in Section~\ref{section:BNF}), but important additional piece of information that BNFs carry are the so-called non-emptiness witnesses~\cite{DBLP:conf/esop/Blanchette0T15}.
These are constants $\const{w} \hastype \alpha_{i_1} \fun \cdots \fun \alpha_{i_k} \fun \seq{\alpha}\;F$, where $I_\const{w} = \{i_1, \ldots, i_k\} \subseteq \{1, \ldots, n\}$, which capture the information which atoms must be given to construct a value of $\seq{\alpha}\;F$.
The constants are subject to the following property for all $i$:
\[
  a_i \in \const{set}_{F,i}\;(\const{w}\;b_{i_1}\cdots \;b_{i_k}) \implies i \in I_\const{w} \land a_i = b_i.
\]
For example, the product type has a single witness $\const{Pair} \hastype \alpha \fun \beta \fun \alpha \times \beta$, whereas the sum type has two: $\const{Inl} \hastype \alpha \fun \alpha + \beta$ and $\const{Inr} \hastype \beta \fun \alpha + \beta$. Non-emptiness witnesses are used to prove that datatypes are non-empty, a requirement for introducing new types in HOL~\cite{Paulson2006TCL}. The fewer arguments a witness has, the more useful it is when proving non-emptiness. Hence, a witness with the (index) set of arguments $I$ subsumes another with the set $J$ if $I \subseteq J$. BNFs carry a complete set of witnesses that are minimal with respect to subsumption.
%\fxnote{AL: I find this section and paragraph very technical. We should either explain more where the non-emptiness witnesses are needed (e.g., a variation on the system from the quotient background section as a datatype without Stop?) or maybe move it to the report version?}

The (co)datatype commands also automatically produce a complete set of witnesses for the resulting types. For example, the list type has a single witness $[] \hastype \alpha\;\tc{list}$.
It forms a complete set because it takes no arguments and thus subsumes every other possible witness.
Similarly, the coinductive lists from Example~\ref{ex:tllist} have a single witness $\const{LNil} \hastype \alpha\;\tc{llist}$.

\looseness=-1
Witnesses can be lifted from the raw type to the quotient type because
\[
  \const{set}_{Q,i}\;[x]_\sim = \bigcap\nolimits_{y \in [\const{map}_F\;\seq{\embedding}\;x]_\sim} \{a \mid \embedding\;a \in \const{set}_{F,i}\;y\} \subseteq \const{set}_{F,i}\;x
\]
for all $x$.
However, the set of witnesses obtained by lifting may stop being complete.
For example, lifting from $\alpha\;\tc{llist} \times \beta$ to terminated lazy lists results in a single witness $\const{tlnil}\hastype\beta \fun (\alpha,\beta)\;\tc{tllist}$, which corresponds to the witness $\lambda b.\;(\const{LNil}, b)$ on the underlying type.
Yet there is a second witness $\const{tlconst}\hastype\alpha \fun (\alpha,\beta)\;\tc{tllist}$ lifted from $\lambda a.\;(\const{lconst}\;a, \const{undef})$,
where the function $\const{lconst}$ is defined corecursively by $\const{lconst}\;a=\const{LCons}\;a\;(\const{lconst}\;a)$ and $\const{undef}$ is an unspecified value of type $\beta$.
The function $f = (\lambda a.\;(\const{lconst}\;a,\const{undef}))$ is not a witness on the raw type because it violates the witness property:
$f\;a$ always contains $\const{undef}$ as an atom of type $\beta$, yet $f$ does not take an argument corresponding to that type.
In contrast, there is no such atom in $\const{tlconst}\;a$, as we will prove in Example~\ref{ex:tllist:transfer}.
Identifying such additional witnesses arising from the equivalence relation's influence on the setters requires manual proofs.

\subsection{Partial quotients}%
\label{section:partial:quotient}

So far, we have focused on total quotients generated by an equivalence relation.
If the relation $\optparen{\sim}$ is only a partial equivalence relation, i.e., symmetric and transitive but not necessarily reflexive, then the resulting quotient is a partial quotient.
Every partial quotient $\seq{\alpha}\;Q = \seq{\alpha}\;F / \quotparen{\sim}$ factors into a subtype $\seq{\alpha}\;T$ of $\seq{\alpha}\;F$ and a total quotient $\seq{\alpha}\;T / \quotparen{\sim'}$.

We can therefore combine the conditions for total quotients with those for subtypes.
Let $\const{Field}_\sim$ denote the field of the symmetric relation $\optparen{\sim}$, i.e. $\const{Field}_\sim = \{ x \mid \exists y.\; x \sim y \}$.
Define $\seq{\alpha}\;T$ as isomorphic to $\const{Field}_\sim$.
The partial equivalence relation $\optparen{\sim}$ on $\seq{\alpha}\;F$ yields an equivalence relation $\optparen{\sim'}$ on $\seq{\alpha}\;T$.
Clearly, $\seq{\alpha}\;Q$ is isomorphic to $\seq{\alpha}\;T / \quotparen{\sim'}$.

The subtype $\seq{\alpha}\;T$ inherits the BNF structure from $\seq{\alpha}\;F$ under two conditions~\cite{Biendarra2015BA,LochbihlerSchneider2018ITP}:
\begin{enumerate}
\item\label{item:subtype:closed}
  The set $\const{Field}_\sim$ must be closed under the mapper $\const{map}_F$, i.e., $\const{map}_F\;\seq{f}\;x \in \const{Field}_\sim$ for all $\seq{f}$ and $x \in \const{Field}_\sim$.

\item\label{item:subtype:reflect}
  Whenever $\const{map}_F\;\seq{\fst}\;z \in \const{Field}_\sim$ and $\const{map}_F\;\seq{\snd}\;z \in \const{Field}_\sim$, then there exists a $y \in \const{Field}_\sim$ with $\forall i.\,\const{set}_{F,i}\; y \subseteq \const{set}_{F,i}\;z$ and $\const{map}_F\;\seq{\fst}\;y = \const{map}_F\;\seq{\fst}\;z$ and $\const{map}_F\;\seq{\snd}\;y = \const{map}_F\;\seq{\snd}\;z$.
\end{enumerate}
The first condition follows from $\const{map}_F$ preserving $\optparen{\sim}$ by the definition of $\const{Field}_\sim$.
The second condition does not follow from our conditions for total quotients.
It ensures that there is a suitable witness $y$ for \theorem{in\_rel} to relate two values in the lifted relation $\const{rel}_F\;\seq{R}$.

\begin{exa}%
  \label{ex:fae}
  In Example~\ref{ex:ae}, the quotient $\alpha\ F_{\tc{seq}} / \quotparen{\sim_{\tc{ae}}}$ of sequences $\alpha\ F_{\tc{seq}} = \tc{nat} \fun \alpha$ does not inherit the BNF structure because $\sim_{\tc{ae}}$ does not satisfy~\eqref{eq:wide:intersections}.
  We now restrict the equivalence relation to sequences that contain only finitely many different values.
  That is, the partial equivalence relation $\sim_{\tc{fae}}$ relates $x$ to $y$ iff
  $\{n \mid x\;n \neq y\;n\}$ and $\const{range}\;x$ and $\const{range}\;y$ are finite, where $\const{range}\;f = f\langle \const{UNIV}\rangle$ denotes the range of $f$.
  So $\const{Field}_{\sim_{\tc{fae}}} = \{ x \mid \const{finite}\ (\const{range}\ x) \}$.

  By the two-step approach, we first define the functor $\alpha\ F_{\tc{fseq}}$ of finitely-valued sequences as the subtype of $F_{\tc{seq}}$.
  As $\const{Field}_{\sim_{\tc{fae}}}$ satisfies the above subtype conditions, $F_{\tc{fseq}}$ inherits the BNF structure from $F_{\tc{seq}}$.
  Second, we define the functor $F_{\tc{fae}}$ of finitely-valued infinitely-different sequences as the total quotient of $F_{\tc{fseq}}$ over $\sim_{\tc{fae}}$.
  By Theorem~\ref{thm:bnf:quotient}, $F_{\tc{fae}}$ inherits the BNF structure;
  wide intersections~\eqref{eq:wide:intersections} are trivially preserved by Lemma~\ref{lemma:wide:intersection:finite} as sequences in $F_{\tc{fseq}}$ contain only finitely many values.
  \exampleend
\end{exa}

Inheriting the BNF structure through the subtype $\seq{\alpha}\;T$ works in many cases.
The next example shows that this may fail in pathological cases though.
So it might be worthwhile to generalize the constructions from the previous sections directly to partial equivalene relations.
This is left as future work.

\begin{exa}%
  \label{ex:partial:incomplete}
  Let the partial equivalence relation $\sim$ relate two finite lists $\mathit{xs}$ and $\mathit{ys}$ iff they contain the same elements, i.e., $\const{set}\ \mathit{xs} = \const{set}\ \mathit{ys}$, and each contains at least one element twice.
  So $\const{Field}_\sim$ consists of all lists that contain at least one element twice.
  Let $\alpha\;T$ be the corresponding subtype of lists.
  Then $\alpha\;T$ does not inherit the BNF structure from lists because condition~\ref{item:subtype:reflect} from above does not hold.
  For the list $\mathit{zs} = [(1,a),(1,b),(2,b)]$, we have $\const{map}\ \const{fst}\ \mathit{zs} = [1,1,2] \in \const{Field}_\sim$ and $\const{map}\ \const{snd}\ \mathit{zs} = [a,b,b] \in \const{Field}_\sim$,
  yet $\mathit{zs}$ is the only list that projects to these two lists and $\mathit{zs} \notin \const{Field}_\sim$.
  This failure is not just because condition~\ref{item:subtype:reflect} is too strong;
  with the mapper for lists, $T$ does not satisfy \theorem{rel\_comp} for the relator defined via \theorem{in\_rel}.
  Nevertheless, the quotient $T / \quotparen{\sim}$ with the inherited mapper becomes the type of non-empty finite sets, which does satisfy the BNF properties.
  \exampleend
\end{exa}

%\end{rep}

\section{Implementation}%
\label{section:implementation}

We provide an Isabelle/HOL command that automatically lifts the BNF structure to total quotient types.
The command requires the user to discharge our conditions on the equivalence relation.
Upon success, it defines the mapper, setters, the relator, and lifted non-emptiness witnesses, and proves the BNF axioms and transfer rules.
The constants' definitions and transfer rules are described in more detail below.
Eventually, the command registers the quotient type with the BNF infrastructure for use in future (co)datatype definitions.
The command was implemented in 1590 lines of Isabelle/ML\@.
All automated proofs are checked by Isabelle's kernel.
Support for partial quotients is left for future work.

\subsection{The \keyw{lift\_bnf} command}%
\label{section:command}

Our implementation extends the interface of the existing \keyw{lift\_bnf} command for subtypes~\cite{Biendarra2015BA}.
Given a quotient type $\seq{\alpha}\;Q = \seq{\alpha}\;F / \quotparen{\sim}$,
\begin{quotex}
  $\keyw{lift\_bnf }\seq{\alpha}\;Q \keyw{ [wits:}\;\seq{w}\keyw{]}$ % chktex 9
\end{quotex}
asks the user to prove the conditions~\eqref{eq:wide:intersections} and~\eqref{eq:weak:pullback} of Theorem~\ref{thm:bnf:quotient}, where~\eqref{eq:wide:intersections} is expressed in terms of~\eqref{eq:wide:intersection:one} according to Lemma~\ref{lem:wide:intersections:alt}.
Since the quotient construction already requires that $\optparen{\sim}$ be an equivalence relation, the remaining condition~\eqref{eq:equiv:sim} holds trivially.
The user can provide optional non-emptiness witnesses $\seq{w}$ as functions mapping to the raw type $\seq{\alpha}\;F$, which adds the corresponding proof obligations (see Section~\ref{section:witness}).

\begin{exa}[Example~\ref{ex:dlist} continued]
  Distinct lists, when viewed as a quotient of lists by the relation $\sim_{\tc{dlist}}$, can be registered as a BNF using $\keyw{lift\_bnf }\alpha\;\tc{dlist}$.
  We do not specify additional witnesses because the list witness $[]$ can be lifted and is already as general as possible.
  Two proof obligations are presented to the user:
  \begin{gather*}
    \forall R\;S.\; R \relcomp S \neq \{\} \implies
      \const{rel}_{\tc{list}}\;R \relcomp \optparen{\sim_{\tc{dlist}}} \relcomp \const{rel}_{\tc{list}}\;S \subseteq
      \optparen{\sim_{\tc{dlist}}} \relcomp \const{rel}_{\tc{list}}\;(R \relcomp S) \relcomp \optparen{\sim_{\tc{dlist}}} \\
    \forall \mathcal{A}.\;\; \mathcal{A} \neq \{\} \land (\bigcap\mathcal{A}) \neq \{\} \implies
      \bigcap\nolimits_{A \in \mathcal{A}} [\{x \mid \const{set}\;x \subseteq A\}]_{\sim_{\tc{dlist}}} \subseteq
      \left[\{x \mid \const{set}\;x \subseteq \bigcap \mathcal{A}\}\right]_{\sim_{\tc{dlist}}}
  \end{gather*}
  The first obligation can be discharged by instantiating Theorem~\ref{thm:confluent:quotient} as described in Example~\ref{ex:dlist}.
  For the second obligation, it suffices to prove that for every list $x$ satisfying
  \begin{equation}\label{eq:dlist:wide-inter:lhs}
    \forall A \in \mathcal{A}.\; \exists y.\; x \sim_{\tc{dlist}} y \land \const{set}\;y \subseteq A
  \end{equation}
  there exists a list $y'$ such that $x \sim_{\tc{dlist}} y'$ and $\const{set}\;y' \subseteq \bigcap\mathcal{A}$.
  We choose $y' = x$.
  Clearly $x \sim_{\tc{dlist}} x$, and since $x \sim_{\tc{dlist}} y$ implies $\const{set}\;x = \const{set}\;y$, it follows from~\eqref{eq:dlist:wide-inter:lhs} that $\const{set}\;x \subseteq A$ for every $A \in \mathcal{A}$.
  Alternatively, the preservation of wide intersections follows directly from Lemma~\ref{lemma:wide:intersection:finite} because $\const{set}\;x$ is finite for every list $x$.
  \exampleend
\end{exa}

After the conditions have been proved by the user, the command defines the BNF constants.
In HOL, the type $\seq{\alpha}\;Q$ is considered distinct from (but isomorphic to) the set of equivalence classes over $\seq{\alpha}\;F$.
Therefore, the definitions use an abstraction function $\const{abs}_Q \hastype \seq{\alpha}\;F \fun \seq{\alpha}\;Q$ and a representation function $\const{rep}_Q \hastype \seq{\alpha}\;Q \fun \seq{\alpha}\;F$ to translate between the types.
Concretely, we define the quotient's mapper by
\begin{equation*}
  \const{map}_Q\;\seq{f} = \const{abs}_Q \circ \const{map}_F\;\seq{f} \circ \const{rep}_Q
\end{equation*}
The quotient's setters use the function $\const{set}_{\unit+{}}$, which maps $\embedding\;a$ to $\{a\}$ and $\circledast$ to $\{\}$:
\begin{equation}
  \label{eq:set:alt:def}
  \const{set}_{Q,i} = \Bigl(\lambda x.\; \bigcap\nolimits_{y \in [\const{map}_F\;\seq{\embedding}\;x]_\sim} \bigcup \const{set}_{\unit+{}}\langle\const{set}_{F,i}\;y\rangle\Bigr) \circ \const{rep}_Q
\end{equation}
This definition is equivalent to the characterization in Theorem~\ref{thm:set:characterization}.
The relator (Theorem~\ref{thm:rel:characterization}) is lifted similarly.
Let $\vimagerel{f}{g}{R}$ denote the inverse image of the relation $R$ under functions $f$ and $g$, i.e., $(x,y) \in \vimagerel{f}{g}{R} \Iff (f\;x, f\;y) \in R$.
Then
\begin{equation}\label{eq:rel:alt:def}
  \const{rel}_Q\;\seq{R} = \vimagerel{\const{rep}_Q}{\const{rep}_Q}{%
    \vimagerel{\const{map}_F\;\seq{\embedding}}{\const{map}_F\;\seq{\embedding}}{%
      \optparen{\sim} \relcomp \const{rel}_F\;\seq{(\const{rel}_{\unit+{}}\;R)} \relcomp \optparen{\sim}}}.
\end{equation}

\subsection{Transfer rule generation}%
\label{section:transfer}

The relationship of a quotient's BNF structure to its underlying type allows us to prove additional properties about the former.
This is achieved by transfer rules, which drive Isabelle's Transfer tool~\cite{HuffmanKuncar2013CPP} (Section~\ref{section:quotient}).
Our command automatically proves parametrized transfer rules for the lifted mapper, setters, and relator.
Parametrized transfer rules are more powerful because they allow the refinement of nested types~\cite[Section~4.3]{Kuncar2016PhD}.
Such rules involve a parametrized correspondence relation $\const{pcr}_Q\;\seq{A} = \const{rel}_F\;\seq{A} \relcomp \const{cr}_Q$, where the parameters $\seq{A}$ relate the type arguments of $F$ and $Q$.

The transfer rule of $\const{map}_Q$ is unsurprising, as it is the canonical lifting of $\const{map}_F$:
\begin{equation*}
  (\const{map}_F, \const{map}_Q) \in \left( \seq{(A \relfun B)} \relfun \const{pcr}_Q\;\seq{A} \relfun \const{pcr}_Q\;\seq{B} \right)
\end{equation*}
%This rule follows from~\eqref{eq:map:respect} and parametricity of $\const{map}_F$, $(\const{map}_F, \const{map}_F) \in (\seq{(A \relfun B)} \relfun \const{rel}_F\;\seq{A} \relfun \const{rel}_F\;\seq{B})$, which holds for all BNFs.
Setters are not transferred to $\const{set}_F$ but to the more complex function from~\eqref{eq:set:alt:def}:
\begin{equation*}
  \Bigl(\lambda x.\; \bigcap\nolimits_{y \in [\const{map}_F\;\seq{\embedding}\;x]_\sim} \bigcup \const{set}_{\unit+{}}\langle\const{set}_{F,i}\;y\rangle,\;\const{set}_{Q,i}\Bigr) \in
  (\const{pcr}_Q\;\seq{A} \relfun \const{rel}_{\tc{set}}\;A_i)
\end{equation*}
where $(X,Y) \in \const{rel}_{\tc{set}}\;A \Iff (\forall x \in X.\;\exists y \in Y.\; (x,y) \in A) \land (\forall y \in Y.\;\exists x \in X.\; (x,y) \in A)$.
Similarly, the rule for $Q$'s relator contains its defining term from~\eqref{eq:rel:alt:def}.

\begin{exa}[Example~\ref{ex:tllist} continued]%
%\begin{exa}
  \label{ex:tllist:transfer}
  Recall that terminated coinductive lists satisfy the conditions for lifting the BNF structure.
  Thus, we obtain the setter $\const{set}_{\tc{tllist},2} \hastype (\alpha,\beta)\;\tc{tllist} \fun \beta\;\tc{set}$ among the other BNF operations.
  We want to prove that $\const{set}_{\tc{tllist},2}\;x$ is empty for all infinite lists $x$.
  To make this precise, let the predicate $\const{lfinite} \hastype \alpha\;\tc{llist} \fun \tc{bool}$ characterize finite coinductive lists.
  We lift it to $(\alpha,\beta)\;\tc{tllist}$ by projecting away the terminator:
  \begin{quotex}
    $\keyw{lift\_definition }\const{tlfinite} \hastype (\alpha,\beta)\;\tc{tllist} \fun \tc{bool}\keyw{ is }(\lambda x.\;\const{lfinite}\;(\fst\;x))$
  \end{quotex}
  Therefore, we have to show that $\forall x.\;\lnot\,\const{tlfinite}\;\var{x} \implies \const{set}_{\tc{tllist},2}\;x = \{\}$.
  Using the transfer rules for the setter and the lifted predicate $\const{tlfinite}$, the \texttt{transfer} proof method reduces the proof obligation to
  \begin{equation*}
    \forall x'.\; \lnot\,\const{lfinite}\;(\fst\;x') \implies \smash{\bigcap\nolimits_{y \in [\const{map}_F\;\seq{\embedding}\;x']_{\sim_\tc{tllist}}}} \bigcup \const{set}_{\unit+{}}\langle\const{set}_{F,2}\;y\rangle = \{\}
  \end{equation*}
  where $x' \hastype (\alpha,\beta)\;F$, and $(\alpha,\beta)\;F = (\alpha\;\tc{llist} \times \beta)$ is the underlying functor of \tc{tllist}.
  The rest of the proof, which need not refer to $\tc{tllist}$ anymore, is automatic.
  As a corollary of this example, $\const{set}_{\tc{tllist},2}\;(\const{tlconst}\;a)$ is always empty, a property we used in Section~\ref{section:witness}.
  \exampleend
\end{exa}

We have also extended \keyw{lift\_bnf} to generate transfer rules for subtypes, for which the setters and relator do not change except for the types.
For example, if $T$ is a subtype of $F$, $\const{set}_{T,i}$ is transferred to $\const{set}_{F,i}$.
Previously no such rules were made available by the command.
This limited the properties that could be proved about the setters and relator to those that follow from the generic BNF axioms.
\begin{exa}[Example~\ref{ex:fae} continued]
  We defined finitely-valued sequences $\alpha\;F_{\tc{fseq}}$ as a subtype of general sequences $\alpha\;F_{\tc{seq}}$.
  Our implementation generates the transfer rule
  \[
    (\const{range}, \const{set}_{\tc{fseq}}) \in (\const{pcr}_{\tc{fseq}}\;A \relfun \const{rel}_{\tc{set}}\;A)
  \]
  for the setter $\const{set}_{\tc{fseq}} \hastype \alpha\;F_{\tc{fseq}} \fun \alpha\;\tc{set}$,
  where $\const{range}$ is the setter of the underlying type $F_{\tc{seq}}$.

\looseness=-1
  We use the above rule to show that $F_{\tc{fseq}}$ deserves its name, namely that $\const{set}_{\tc{fseq}}\;x$ is finite for all $x$.
  Using that $\const{pcr}_{\tc{fseq}}\;(=)$ is equal to $\const{cr}_{\tc{fseq}}$ and $\const{rel}_{\tc{set}}\;(=)$ is the equality relation,
  the transfer rule implies $\const{set}_{\tc{fseq}}\;x = \const{range}\;y$ for all $y$ where $(y,x) \in \const{cr}_{\tc{fseq}}$.
  Recall that the correspondence relation $\const{cr}_{\tc{fseq}}$ relates exactly those sequences that have finite range to their isomorphic copies in $F_{\tc{fseq}}$.
  Hence, every $y$ with $(y,x) \in \const{cr}_{\tc{fseq}}$ must have finite range.
  These reasoning steps are automated by the \texttt{transfer} proof method.
  \exampleend
\end{exa}

%\subsection{Efforts}
%
%- dlist comparison???
%
%- regex example

\section{Related work}%
\label{section:related:work}

Quotient constructions have been formalized and implemented, e.g., in
Isabelle/HOL\cite{HuffmanKuncar2013CPP,KaliszykUrban2011SAC,Paulson2006TCL,Slotosch1997TPHOLs},
HOL4~\cite{Homeier2005TPHOLs},
Agda~\cite{Veltri2015SPLST,Veltri2017phd} (as well as the Cubical Agda variant~\cite{DBLP:journals/jfp/VezzosiMA21,DBLP:conf/fscd/Veltri21}),
Cedille~\cite{MarmadukeJenkinsStump2019TFP},
Coq~\cite{Cohen2013ITP,ChicliPottierSimpson2003TYPES},
Lean~\cite{AvigadCarneiroHudon2019ITP}, and
Nuprl~\cite{Nogin2002TPHOLs}.
None of these works look at the preservation of functor properties except for Avigad et al.~\cite{AvigadCarneiroHudon2019ITP} (discussed in Section~\ref{section:QPF}) and Veltri~\cite{Veltri2017phd,DBLP:conf/fscd/Veltri21}.

Veltri~\cite{Veltri2017phd} studies the special case of when the delay monad is preserved by a quotient of weak bisimilarity, focusing on the challenges that quotients pose in intensional type theory.
Furthermore, he~\cite{DBLP:conf/fscd/Veltri21} examines, using Cubical Agda, how different constructions of the finite powerset functor $\tc{fset}$ affect codatatype recursion through this functor.
He notices that specifying $\tc{fset}$ as a higher-inductive type via an equational presentation works smoothly in a constructive setting.
In contrast, first constructing an intermediate codatatype on the underlying raw type $\tc{list}$ and then quotienting it with the equivalence relation lifted to the codatatype requires the full axiom of choice for deriving the finality theorem.
Analogously, we do not construct an intermediate codatatype and delay the quotienting, but the corecursion directly goes through the quotient BNF thanks to the BNF closure properties.

Abbot et al.~\cite{AbbottAltenkirchGhaniMcBride2004MPC} introduce quotient containers as a model of datatypes with permutative structure, such as unordered pairs, cyclic lists, and multisets.
The map function of quotient containers does not change the shape of the container.
Quotient containers therefore cannot deal with quotients where the equivalence relation takes the identity of elements into account, such as distinct lists, finite sets, and the free idempotent monoid.
Overall our construction strictly subsumes quotient containers.

\subsection{Quotients in the category of Sets}

BNFs are accessible functors in the category of Sets.
We therefore relate to the literature on when quotients preserve functors and their properties in Set.

Trnkov\'a~\cite{Trnkova1969CMUC} showed that all Set functors preserve non-empty intersections:
in our notation $\internalize{F}\;A \cap \internalize{F}\;B = \internalize{F}\;(A \cap B)$ whenever $A \cap B \neq \{\}$.
Empty intersections need not be preserved though.
Functors that do are called regular~\cite{Trnkova1981CMUC} or sound~\cite{AdamekGummTrnkova2010JLC}.
All BNFs are sound as $\internalize{F}\;A = \{ x \mid \const{set}_F\;x \subseteq A \}$.
The naive quotient construction can lead to unsound functors, as shown in Example~\ref{ex:Inl:Inl}.

Every unsound functor can be ``repaired'' by setting $\internalize{F}\;\{\}$ to the \emph{distinguished points} $\const{dp}_F$.
We write $\internalize{F}\,'$ for the repaired action.
\begin{equation}
  \label{eq:repair}
  \internalize{F}\,'\;A =
  \begin{cases}
    \const{dp}_F
    &
    \text{if } A = \{\}
    \\
    \internalize{F}\;A & \text{otherwise}
  \end{cases}
\end{equation}
Trnkov\'a characterizes the distinguished points $\const{dp}_F$ as the natural transformations from $C_{1,0}$ to $F$ where $\internalize{C_{1,0}}\; \{\} = \{\}$ and $\internalize{C_{1,0}}\; A = \{\circledast\}$ for $A \neq \{\}$.
Barr~\cite{Barr1993TCS} and Gumm~\cite{Gumm2005CALCO} use equalizers instead of natural transformations to define the distinguished points of univariate functors:
\begin{equation}
  \label{eq:distinguished:points}
  \const{dp}_F =
  \{ x \mid \const{map}_F\;(\lambda \_.\; \const{True})\;x = \const{map}_F\;(\lambda \_.\; \const{False})\; x \}
\end{equation}
\looseness=-1
The case distinction in~\eqref{eq:repair} makes it hard to work with repaired functors, especially as the case distinctions proliferate for multivariate functors.
Instead, we repair the unsoundness by avoiding empty sets altogether.
Our characterization $\internalizesim{F}{\sim}\;A$ in Lemma~\ref{lem:F_in'} effectively derives the quotient from $(\unit + \alpha)\;F$ instead of $\alpha\;F$.
Moreover, our characterization of $\internalize{F}\;\seq{A}$ generalizes Barr and Gumm's definition of distinguished points: for $\seq{A} = \{\}$,~\eqref{eq:F_in} simplifies to~\eqref{eq:distinguished:points}, using preimage preservation~\eqref{eq:preimage:preservation}.
The resulting quotient is the same because $[\internalizesim{F}{\sim}\;\seq{A}]_\sim = [\internalize{F}\;\seq{A}]_\sim$ if $A_i \neq \{\}$ for all $i$.

Given the other BNF properties, subdistributivity is equivalent to the functor preserving weak pullbacks.
Ad\'amek et al.~\cite{AdamekGummTrnkova2010JLC} showed that an accessible Set functor preserves weak pullbacks iff it has a so-called dominated presentation in terms of flat equations $E$ over a signature $\Sigma$.
This characterization does not immediately help with proving subdistributivity, though.
For example, the finite set quotient $\alpha\;\tc{fset} = \alpha\;\tc{list} / \quotparen{\sim_{\tc{fset}}}$ comes with the signature $\Sigma = \{ \sigma_n \mid n \in \mathbb{N} \}$ and the equations $\sigma_n(x_1,\ldots x_n) = \sigma_m(y_1,\ldots,y_m)$ whenever $\{x_1,\ldots,x_n\} = \{y_1,\ldots,y_m\}$.
Proving domination for this presentation boils down to proving subdistributivity directly.
Our criterion using a confluent relation (Theorem~\ref{thm:confluent:quotient}) is only sufficient, not necessary, but it greatly reduces the actual proof effort.

\subsection{Lean's quotients of polynomial functors}%
\label{section:QPF}

Avigad et al.~\cite{AvigadCarneiroHudon2019ITP} proposed quotients of polynomial functors (QPF) as a model for datatypes.
QPFs generalize BNFs in that they require less structure: there is no setter and the relator need not satisfy subdistributivity.
Nevertheless, the quotient construction is similar to ours.
Without loss of generality, we consider in our comparison only the univariate case $\alpha\;Q = \alpha\;F / \quotparen{\sim}$.

The main difference lies in the definition of the liftings $\const{lift}_F$ of predicates $P \hastype \alpha \fun \tc{bool}$ and relations $R \hastype \alpha \reltype \beta$.
In our notation, $\const{lift}_F\;P$ corresponds to $\lambda x.\; x \in \internalize{F}\;\{ a \mid P\;a \}$ and $\const{lift}_F\;R$ to $\const{rel}_F\; R$.
QPFs define these liftings for the quotient $Q$ as follows:
\begin{equation*}
  \const{lift}_Q\;P\;[x]_\sim = (\exists x' \in [x]_\sim.\; P\; x')
  \qquad
  \const{lift}_Q\;R\;[x]_\sim\;[y]_\sim = (\exists x' \in [x]_\sim.\; \exists y' \in [y]_\sim.\;R\;x'\;y')
\end{equation*}
That is, these definitions correspond to the naive construction $\internalize{Q}\;A = [\internalize{F}\;A]_\sim$ and $\const{rel}_Q\;R = [\const{rel}_F\;R]_\sim$, where $[(x, y)]_\sim = ([x]_\sim, [y]_\sim)$.
As discussed above, the resulting quotient may be an unsound functor.
Consequently, lifting of predicates does not preserve empty intersections in general.
This hinders modular proofs.
For example, suppose that a user has already shown $\const{lift}_Q\;P_1\;x$ and $\const{lift}_Q\;P_2\;x$ for some value $x$ and two properties $P_1$ and $P_2$.
Then, to deduce $\const{lift}_F\;(\lambda a.\;P_1\; a \wedge P_2\;a)\;x$, they would have to prove that the two properties do not contradict each other, i.e., $\exists a.\;P_1\;a \wedge P_2\;a$.
Obviously, this makes modular proofs harder as extra work is needed to combine properties.

QPFs use $\const{lift}_F\;P$ in the induction theorem for datatypes.
So when a datatype recurses through $\tc{tllist}$,
the aforementioned obstacle spreads to proofs by induction:
splitting a complicated inductive statement into smaller lemmas is not for free.
Moreover, $\const{lift}_Q$ holds for fewer values, as the next example shows.
Analogous problems arise in QPFs for relation lifting, which appears in the coinduction theorem.

\begin{exa}[Example~\ref{ex:tllist:transfer} cont.]
  Consider the infinite repetition $\const{tlconst}\;a \hastype (\alpha, \beta)\;\tc{tllist}$ of the atom $a$ as a terminated lazy list.
  As $\const{tlconst}\;a$ contains only $a$s, one would expect that $\const{lift}_\tc{tllist}\;(\lambda a'.\; a' = a)\;(\lambda \_.\;\const{False})\;(\const{tlconst}\;a)$ holds.
  Yet this property is provably false.
  \exampleend
\end{exa}

These issues would go away if $\const{lift}_Q$ was defined following our approach for $\internalize{Q}\;A = [\internalizesim{F}{\sim}\;A]_\sim$ and $\const{rel}_Q$ as in Theorem~\ref{thm:rel:characterization}.
These definitions do not rely on the additional BNF structure;
only $\const{map}_Q$ is needed and QPFs define $\const{map}_Q$ like we do.
The repair should therefore work for the general QPF case as well.

\section{Conclusion}

\looseness=-1
We have described a sufficient criterion for quotient types to be able to inherit the BNF structure from the underlying type.
We have demonstrated the effectiveness of the criterion by automating the BNF ``inheritance'' in the form of the \keyw{lift\_bnf} command in Isabelle/HOL and used it (which amounts to proving the criterion) for several realistic quotient types.
We have also argued that our treatment of the quotient's setter and relator to avoid unsoundness carries over to more general structures, such as Lean's QPFs.

As future work, we plan to investigate quotients of existing generalizations of BNFs to co- and contravariant functors~\cite{LochbihlerSchneider2018ITP} and functors operating on small-support endomorphisms and bijections~\cite{DBLP:journals/pacmpl/BlanchetteGPT19}.
Furthermore, we would like to provide better automation for proving subdistributivity via confluent rewrite systems as part of \keyw{lift\_bnf}.

\def\ackname{\relax Acknowledgment}

%\vspace{-0.8\baselineskip}

\section*{\ackname}
We thank David Basin for supporting this work, Ralf Sasse and Andrei Popescu for insightful discussions about confluent relations, BNFs, their preservation of wide intersections, and ways to express the setters in terms of the mapper, and Jasmin Blanchette and the anonymous IJCAR and LMCS reviewers for numerous comments on earlier drafts of this article, which helped to improve the presentation. Julian Biendarra developed the original \keyw{lift\_bnf} command for subtypes, which we extended to quotient types in this work.

\bibliographystyle{alphaurl}
\bibliography{lmcs.bib}

\newcommand{\etalchar}[1]{$^{#1}$}
\begin{thebibliography}{AAGM04}

\bibitem[AAGM04]{AbbottAltenkirchGhaniMcBride2004MPC}
Michael Abbott, Thorsten Altenkirch, Neil Ghani, and Conor McBride.
\newblock Constructing polymorphic programs with quotient types.
\newblock In Dexter Kozen, editor, {\em {MPC} 2004}, volume 3125 of {\em LNCS},
  pages 2--15, Berlin, Heidelberg, 2004. Springer.

\bibitem[ACH19]{AvigadCarneiroHudon2019ITP}
Jeremy Avigad, Mario Carneiro, and Simon Hudon.
\newblock Data types as quotients of polynomial functors.
\newblock In John Harrison, John O'Leary, and Andrew Tolmach, editors, {\em ITP
  2019}, volume 141 of {\em LIPIcs}, pages 6:1--6:19. Schloss Dagstuhl -
  Leibniz-Zentrum f{\"{u}}r Informatik, 2019.
\newblock \href {https://doi.org/10.4230/LIPIcs.ITP.2019.6}
  {\path{doi:10.4230/LIPIcs.ITP.2019.6}}.

\bibitem[AGT10]{AdamekGummTrnkova2010JLC}
Jir{\'{\i}} Ad{\'a}mek, H.~Peter Gumm, and Vera Trnkov{\'a}.
\newblock Presentation of set functors: A coalgebraic perspective.
\newblock {\em J. Log. Comput.}, 20(5):991--1015, 2010.
\newblock \href {https://doi.org/10.1093/logcom/exn090}
  {\path{doi:10.1093/logcom/exn090}}.

\bibitem[Bar93]{Barr1993TCS}
Michael Barr.
\newblock Terminal coalgebras in well-founded set theory.
\newblock {\em Theor. Comput. Sci.}, 114(2):299--315, 1993.
\newblock \href {https://doi.org/10.1016/0304-3975(93)90076-6}
  {\path{doi:10.1016/0304-3975(93)90076-6}}.

\bibitem[BGPT19]{DBLP:journals/pacmpl/BlanchetteGPT19}
Jasmin~Christian Blanchette, Lorenzo Gheri, Andrei Popescu, and Dmitriy
  Traytel.
\newblock Bindings as bounded natural functors.
\newblock {\em {PACMPL}}, 3({POPL}):22:1--22:34, 2019.
\newblock \href {https://doi.org/10.1145/3290335} {\path{doi:10.1145/3290335}}.

\bibitem[BHL{\etalchar{+}}14]{blanchette14itp}
Jasmin~Christian Blanchette, Johannes H{\"o}lzl, Andreas Lochbihler, Lorenz
  Panny, Andrei Popescu, and Dmitriy Traytel.
\newblock Truly modular (co)datatypes for {Isabelle/HOL}.
\newblock In Gerwin Klein and Ruben Gamboa, editors, {\em ITP 2014}, volume
  8558 of {\em LNCS}, pages 93--110. Springer, 2014.

\bibitem[Bie15]{Biendarra2015BA}
Julian Biendarra.
\newblock Functor-preserving type definitions in {Isabelle/HOL}.
\newblock Bachelor thesis, Fakult\"at f\"ur Informatik, Technische
  Universit\"at M\"unchen, 2015.

\bibitem[BKT17]{DBLP:conf/rv/BasinKT17}
David~A. Basin, Srdan Krstic, and Dmitriy Traytel.
\newblock Almost event-rate independent monitoring of metric dynamic logic.
\newblock In Shuvendu~K. Lahiri and Giles Reger, editors, {\em {RV} 2017},
  volume 10548 of {\em LNCS}, pages 85--102. Springer, 2017.
\newblock \href {https://doi.org/10.1007/978-3-319-67531-2_6}
  {\path{doi:10.1007/978-3-319-67531-2_6}}.

\bibitem[BLS20]{BasinLochbihlerSefidgar2020JC}
David~A. Basin, Andreas Lochbihler, and S.~Reza Sefidgar.
\newblock {CryptHOL}: Game-based proofs in higher-order logic.
\newblock {\em J. Cryptology}, 33:494--566, 2020.
\newblock \href {https://doi.org/10.1007/s00145-019-09341-z}
  {\path{doi:10.1007/s00145-019-09341-z}}.

\bibitem[BPT14]{DBLP:conf/itp/Blanchette0T14}
Jasmin~Christian Blanchette, Andrei Popescu, and Dmitriy Traytel.
\newblock Cardinals in {Isabelle/HOL}.
\newblock In Gerwin Klein and Ruben Gamboa, editors, {\em {ITP} 2014}, volume
  8558 of {\em LNCS}, pages 111--127. Springer, 2014.
\newblock \href {https://doi.org/10.1007/978-3-319-08970-6_8}
  {\path{doi:10.1007/978-3-319-08970-6_8}}.

\bibitem[BPT15]{DBLP:conf/esop/Blanchette0T15}
Jasmin~Christian Blanchette, Andrei Popescu, and Dmitriy Traytel.
\newblock Witnessing (co)datatypes.
\newblock In Jan Vitek, editor, {\em {ESOP} 2015}, volume 9032 of {\em LNCS},
  pages 359--382. Springer, 2015.
\newblock \href {https://doi.org/10.1007/978-3-662-46669-8_15}
  {\path{doi:10.1007/978-3-662-46669-8_15}}.

\bibitem[CDM13]{CohenDenesMortberg2013CPP}
Cyril Cohen, Maxime D{\'e}n{\`e}s, and Anders M{\"o}rtberg.
\newblock Refinements for free!
\newblock In Georges Gonthier and Michael Norrish, editors, {\em CPP 2013},
  volume 8307 of {\em LNCS}, pages 147--162. Springer, 2013.

\bibitem[Coh13]{Cohen2013ITP}
Cyril Cohen.
\newblock Pragmatic quotient types in {Coq}.
\newblock In Sandrine Blazy, Christine Paulin-Mohring, and David Pichardie,
  editors, {\em ITP 2013}, volume 7889 of {\em LNCS}, pages 213--228. Springer,
  2013.
\newblock \href {https://doi.org/10.1007/978-3-642-39634-2_17}
  {\path{doi:10.1007/978-3-642-39634-2_17}}.

\bibitem[CPS03]{ChicliPottierSimpson2003TYPES}
Laurent Chicli, Lo{\"i}c Pottier, and Carlos Simpson.
\newblock Mathematical quotients and quotient types in {Coq}.
\newblock In Herman Geuvers and Freek Wiedijk, editors, {\em TYPES 2003},
  volume 2646 of {\em LNCS}, pages 95--107. Springer, 2003.
\newblock \href {https://doi.org/10.1007/3-540-39185-1_6}
  {\path{doi:10.1007/3-540-39185-1_6}}.

\bibitem[DV13]{DBLP:conf/ijcai/GiacomoV13}
Giuseppe {De Giacomo} and Moshe~Y. Vardi.
\newblock Linear temporal logic and linear dynamic logic on finite traces.
\newblock In Francesca Rossi, editor, {\em {IJCAI} 2013}, pages 854--860.
  {IJCAI/AAAI}, 2013.

\bibitem[FLST20]{DBLP:conf/cade/FurerLST20}
Basil F{\"{u}}rer, Andreas Lochbihler, Joshua Schneider, and Dmitriy Traytel.
\newblock Quotients of bounded natural functors.
\newblock In Nicolas Peltier and Viorica Sofronie{-}Stokkermans, editors, {\em
  {IJCAR} 2020}, volume 12167 of {\em LNCS}, pages 58--78. Springer, 2020.
\newblock \href {https://doi.org/10.1007/978-3-030-51054-1_4}
  {\path{doi:10.1007/978-3-030-51054-1_4}}.

\bibitem[GS05]{GummSchroeder2005AU}
H.~Peter Gumm and Tobias Schr{\"o}der.
\newblock Types and coalgebraic structure.
\newblock {\em Algebra universalis}, 53(2):229--252, 2005.
\newblock \href {https://doi.org/10.1007/s00012-005-1888-2}
  {\path{doi:10.1007/s00012-005-1888-2}}.

\bibitem[Gum05]{Gumm2005CALCO}
H.~Peter Gumm.
\newblock From {T}-coalgebras to filter structures and transition systems.
\newblock In Jos{\'e}~Luiz Fiadeiro, Neil Harman, Markus Roggenbach, and Jan
  Rutten, editors, {\em Algebra and Coalgebra in Computer Science}, volume 3629
  of {\em LNCS}, pages 194--212, Berlin, Heidelberg, 2005. Springer.

\bibitem[HK13]{HuffmanKuncar2013CPP}
Brian Huffman and Ond{\v{r}}ej Kun{\v{c}}ar.
\newblock Lifting and {Transfer}: {A} modular design for quotients in
  {I}sabelle/{HOL}.
\newblock In Georges Gonthier and Michael Norrish, editors, {\em CPP 2013},
  volume 8307 of {\em LNCS}, pages 131--146. Springer, 2013.
\newblock \href {https://doi.org/10.1007/978-3-319-03545-1_9}
  {\path{doi:10.1007/978-3-319-03545-1_9}}.

\bibitem[Hom05]{Homeier2005TPHOLs}
Peter~V. Homeier.
\newblock A design structure for higher order quotients.
\newblock In Joe Hurd and Tom Melham, editors, {\em TPHOLs 2005}, volume 3603
  of {\em LNCS}, pages 130--146, Berlin, Heidelberg, 2005. Springer.
\newblock \href {https://doi.org/10.1007/11541868_9}
  {\path{doi:10.1007/11541868_9}}.

\bibitem[Hul80]{Hullot1980SRI}
Jean-Marie Hullot.
\newblock A catalogue of canonical term rewrite systems.
\newblock Technical Report CSL-113, SRI International, 1980.

\bibitem[KP19]{KuncarPopescu2019JAR}
Ond{\v{r}}ej Kun{\v{c}}ar and Andrei Popescu.
\newblock From types to sets by local type definition in higher-order logic.
\newblock {\em J. Autom. Reasoning}, 62(2):237--260, 2019.
\newblock \href {https://doi.org/10.1007/s10817-018-9464-6}
  {\path{doi:10.1007/s10817-018-9464-6}}.

\bibitem[KU11]{KaliszykUrban2011SAC}
Cezary Kaliszyk and Christian Urban.
\newblock Quotients revisited for {Isabelle/HOL}.
\newblock In William~C. Chu, W.~Eric Wong, Mathew~J. Palakal, and Chih{-}Cheng
  Hung, editors, {\em SAC 2011}, pages 1639--1644. ACM, 2011.
\newblock \href {https://doi.org/10.1145/1982185.1982529}
  {\path{doi:10.1145/1982185.1982529}}.

\bibitem[Kun16]{Kuncar2016PhD}
Ond{\v{r}}ej Kun{\v{c}}ar.
\newblock {\em Types, Abstraction and Parametric Polymorphism in Higher-Order
  Logic}.
\newblock PhD thesis, Technical University Munich, Germany, 2016.

\bibitem[LL19]{LammichLochbihler2018JAR}
Peter Lammich and Andreas Lochbihler.
\newblock Automatic refinement to efficient data structures: A comparison of
  two approaches.
\newblock {\em J. Autom. Reasoning}, 63(1):53--94, 2019.
\newblock \href {https://doi.org/10.1007/s10817-018-9461-9}
  {\path{doi:10.1007/s10817-018-9461-9}}.

\bibitem[Loc10]{Coinductive2013}
Andreas Lochbihler.
\newblock Coinductive.
\newblock {\em Archive of Formal Proofs}, 2010.
\newblock \url{http://isa-afp.org/entries/Coinductive.html}, Formal proof
  development.

\bibitem[Loc19]{Lochbihler2019jar}
Andreas Lochbihler.
\newblock Effect polymorphism in higher-order logic (proof pearl).
\newblock {\em J. Autom. Reasoning}, 63(2):439--462, 2019.
\newblock \href {https://doi.org/10.1007/s10817-018-9476-2}
  {\path{doi:10.1007/s10817-018-9476-2}}.

\bibitem[LS16]{LochbihlerSchneider2016ITP}
Andreas Lochbihler and Joshua Schneider.
\newblock Equational reasoning with applicative functors.
\newblock In Jasmin~Christian Blanchette and Stephan Merz, editors, {\em ITP
  2016}, volume 9807 of {\em LNCS}, pages 252--273. Springer, 2016.
\newblock \href {https://doi.org/10.1007/978-3-319-43144-4_16}
  {\path{doi:10.1007/978-3-319-43144-4_16}}.

\bibitem[LS18]{LochbihlerSchneider2018ITP}
Andreas Lochbihler and Joshua Schneider.
\newblock Relational parametricity and quotient preservation for modular
  (co)datatypes.
\newblock In Jeremy Avigad and Assia Mahboubi, editors, {\em ITP 2018}, volume
  10895 of {\em LNCS}, pages 411--431. Springer, 2018.
\newblock \href {https://doi.org/10.1007/978-3-319-94821-8_24}
  {\path{doi:10.1007/978-3-319-94821-8_24}}.

\bibitem[LSBM19]{LochbihlerSefidgarBasinMaurer2019CSF}
Andreas Lochbihler, S.~Reza Sefidgar, David~A. Basin, and Ueli Maurer.
\newblock Formalizing constructive cryptography using {CryptHOL}.
\newblock In {\em CSF 2019}, pages 152--166. IEEE, 2019.
\newblock \href {https://doi.org/10.1109/CSF.2019.00018}
  {\path{doi:10.1109/CSF.2019.00018}}.

\bibitem[MJS19]{MarmadukeJenkinsStump2019TFP}
Andrew Marmaduke, Christopher Jenkins, and Aaron Stump.
\newblock Quotient types by normalization in {Cedille}.
\newblock In {\em TFP 2019}, 2019.

\bibitem[NK14]{DBLP:books/sp/NipkowK14}
Tobias Nipkow and Gerwin Klein.
\newblock {\em Concrete Semantics -- With {Isabelle/HOL}}.
\newblock Springer, 2014.
\newblock \href {https://doi.org/10.1007/978-3-319-10542-0}
  {\path{doi:10.1007/978-3-319-10542-0}}.

\bibitem[Nog02]{Nogin2002TPHOLs}
Aleksey Nogin.
\newblock Quotient types: A modular approach.
\newblock In Victor~A. Carre{\~{n}}o, C{\'e}sar~A. Mu{\~{n}}oz, and Sofi{\`e}ne
  Tahar, editors, {\em TPHOLs 2002}, volume 2410 of {\em LNCS}, pages 263--280.
  Springer, 2002.
\newblock \href {https://doi.org/10.1007/3-540-45685-6_18}
  {\path{doi:10.1007/3-540-45685-6_18}}.

\bibitem[NT14]{DBLP:conf/itp/NipkowT14}
Tobias Nipkow and Dmitriy Traytel.
\newblock Unified decision procedures for regular expression equivalence.
\newblock In Gerwin Klein and Ruben Gamboa, editors, {\em {ITP} 2014}, volume
  8558 of {\em LNCS}, pages 450--466. Springer, 2014.
\newblock \href {https://doi.org/10.1007/978-3-319-08970-6_29}
  {\path{doi:10.1007/978-3-319-08970-6_29}}.

\bibitem[Pau06]{Paulson2006TCL}
Lawrence~C. Paulson.
\newblock Defining functions on equivalence classes.
\newblock {\em ACM Trans. Comput. Logic}, 7(4):658--675, 2006.
\newblock \href {https://doi.org/10.1145/1183278.1183280}
  {\path{doi:10.1145/1183278.1183280}}.

\bibitem[Slo97]{Slotosch1997TPHOLs}
Oscar Slotosch.
\newblock Higher order quotients and their implementation in {Isabelle/HOL}.
\newblock In Elsa~L. Gunter and Amy Felty, editors, {\em TPHOLs 1997}, volume
  1275 of {\em LNCS}, pages 291--306, Berlin, Heidelberg, 1997. Springer.
\newblock \href {https://doi.org/10.1007/BFb0028401}
  {\path{doi:10.1007/BFb0028401}}.

\bibitem[Soz10]{Sozeau2010JFR}
Matthieu Sozeau.
\newblock A new look at generalized rewriting in type theory.
\newblock {\em J. Formalized Reasoning}, 2(1):41--62, 2010.
\newblock \href {https://doi.org/10.6092/issn.1972-5787/1574}
  {\path{doi:10.6092/issn.1972-5787/1574}}.

\bibitem[TPB12]{TraytelPopescuBlanchette2012LICS}
Dmitriy Traytel, Andrei Popescu, and Jasmin~Christian Blanchette.
\newblock Foundational, compositional (co)datatypes for higher-order logic:
  Category theory applied to theorem proving.
\newblock In {\em LICS 2012}, pages 596--605. IEEE Computer Society, 2012.
\newblock \href {https://doi.org/10.1109/LICS.2012.75}
  {\path{doi:10.1109/LICS.2012.75}}.

\bibitem[Trn69]{Trnkova1969CMUC}
V{\v{e}}ra Trnkov{\'a}.
\newblock Some properties of set functors.
\newblock {\em Commentationes Mathematicae Universitatis Carolinae},
  10(2):323--352, 1969.

\bibitem[Trn71]{Trnkova1981CMUC}
V{\v{e}}ra Trnkov{\'a}.
\newblock On descriptive classification of set-functors {I}.
\newblock {\em Commentationes Mathematicae Universitatis Carolinae},
  12(1):143--174, 1971.

\bibitem[Vel15]{Veltri2015SPLST}
Niccol{\`o} Veltri.
\newblock Two set-based implementations of quotients in type theory.
\newblock In J.~Nummenmaa, O.~Sievi-Korte, and E.~M{\"a}kinen, editors, {\em
  SPLST 2015}, volume 1525 of {\em CEUR Workshop Proceedings}, pages 194--205,
  2015.

\bibitem[Vel17]{Veltri2017phd}
Niccol{\`o} Veltri.
\newblock {\em A Type-Theoretical Study of Nontermination}.
\newblock PhD thesis, Tallinn University of Technology, 2017.

\bibitem[Vel21]{DBLP:conf/fscd/Veltri21}
Niccol{\`{o}} Veltri.
\newblock Type-theoretic constructions of the final coalgebra of the finite
  powerset functor.
\newblock In Naoki Kobayashi, editor, {\em {FSCD} 2021}, volume 195 of {\em
  LIPIcs}, pages 22:1--22:18. Schloss Dagstuhl - Leibniz-Zentrum f{\"{u}}r
  Informatik, 2021.
\newblock \href {https://doi.org/10.4230/LIPIcs.FSCD.2021.22}
  {\path{doi:10.4230/LIPIcs.FSCD.2021.22}}.

\bibitem[VMA21]{DBLP:journals/jfp/VezzosiMA21}
Andrea Vezzosi, Anders M{\"{o}}rtberg, and Andreas Abel.
\newblock Cubical agda: {A} dependently typed programming language with
  univalence and higher inductive types.
\newblock {\em J. Funct. Program.}, 31:e8, 2021.
\newblock \href {https://doi.org/10.1017/S0956796821000034}
  {\path{doi:10.1017/S0956796821000034}}.

\end{thebibliography}

\end{document}